\documentclass[12pt,english,12pt,draftclsnofoot,english,onecolumn]{IEEEtran}
\usepackage[T1]{fontenc}
\usepackage{babel}
\usepackage{array}
\usepackage{multirow}
\usepackage{amsmath}
\usepackage{amssymb}
\usepackage{graphicx}
\usepackage{setspace}
\usepackage[unicode=true,
 bookmarks=true,bookmarksnumbered=true,bookmarksopen=true,bookmarksopenlevel=1,
 breaklinks=false,pdfborder={0 0 0},backref=false,colorlinks=false]
 {hyperref}
\hypersetup{pdftitle={Your Title},
 pdfauthor={Your Name},
 pdfpagelayout=OneColumn, pdfnewwindow=true, pdfstartview=XYZ, plainpages=false}

\makeatletter

\providecommand{\tabularnewline}{\\}

\ifCLASSOPTIONcompsoc
\usepackage[caption=false,font=normalsize,labelfont=sf,textfont=sf]{subfig}
\else
\usepackage[caption=false,font=footnotesize]{subfig}
\fi
\usepackage{cite}
\usepackage{amsthm}
\newtheorem{theorem}{\textbf{Theorem}}
\newtheorem{corollary}{\textbf{Corollary}}
\newtheorem{lemma}{\textbf{Lemma}}
\newtheorem{proposition}{\textbf{Proposition}}
\usepackage{algorithm}
\usepackage{amsmath}

\@ifundefined{showcaptionsetup}{}{%
 \PassOptionsToPackage{caption=false}{subfig}}
\usepackage{subfig}
\makeatother

\begin{document}

\title{Modeling and Analysis of SCMA Enhanced D2D and Cellular Hybrid Network}

\author{{\small{\IEEEauthorblockN{Junyu Liu, Min Sheng, Lei Liu, Yan Shi
and Jiandong Li}}}\\
{\small{\IEEEauthorblockA{State Key Laboratory of Integrated Service
Networks, Xidian University, Xi'an, Shaanxi, 710071, China}}\\
{\small{Email: \{jyliu,leiliu\}@stu.xidian.edu.cn, \{msheng,jdli\}@mail.xidian.edu.cn,
yshi@xidian.edu.cn}}}}
\maketitle
\begin{abstract}
Sparse code multiple access (SCMA) has been recently proposed for
the future wireless networks, which allows non-orthogonal spectrum
resource sharing and enables system overloading. In this paper, we
apply SCMA into device-to-device (D2D) communication and cellular
hybrid network, targeting at using the overload feature of SCMA to
support massive device connectivity and expand network capacity. Particularly,
we develop a stochastic geometry based framework to model and analyze
SCMA, considering underlaid and overlaid mode. Based on the results,
we analytically compare SCMA with orthogonal frequency-division multiple
access (OFDMA) using area spectral efficiency (ASE) and quantify closed-form
ASE gain of SCMA over OFDMA. Notably, it is shown that system ASE
can be significantly improved using SCMA and the ASE gain scales linearly
with the SCMA codeword dimension. Besides, we endow D2D users with
an activated probability to balance cross-tier interference in the
underlaid mode and derive the optimal activated probability. Meanwhile,
we study resource allocation in the overlaid mode and obtain the optimal
codebook allocation rule. It is interestingly found that the optimal
SCMA codebook allocation rule is independent of cellular network parameters
when cellular users are densely deployed. The results are helpful
in the implementation of SCMA in the hybrid system.
\end{abstract}

\section{Introduction\label{sec:Introduction}}

\IEEEPARstart{T}{he} fifth generation (5G) wireless networks are
expected to be a mixture of network architectures with a number of
ambitious goals, including exploding mobile data volume (1000$\times$
improvement over 2010 by 2020 \cite{2011_Cisco}) and massive connected
devices (around 50 billion by 2020 \cite{White_paper_50}), etc. Device-to-device
(D2D) communication and cellular hybrid network is a promising architecture
to achieve these goals \cite{D2D_Magazine_ref}. In particular, due
to the proximity feature of D2D communication devices, spatial resources
can be effectively exploited, thereby improving spectrum utilization
and network capacity. Moreover, since direct local transmission is
enabled to bypass the cellular base station (BS), considerable traffic
can be offloaded from BSs. Hence, huge number of devices can be simultaneously
supported.

The coexistence of cellular network and D2D network can be generally
categorized into underlaid mode and overlaid mode \cite{D2D_survey,D2D_Magazine_ref1,Spectrum_sharing_ref,TCR_ref}.
In the underlaid mode, cellular users and D2D users simultaneously
transmit data over the same resources (e.g., spectrum resources).
In consequence, cross-tier interference exists between cellular users
and D2D users. In order to combat the overwhelming interference, efficient
resource sharing and interference management mechanisms have been
proposed in \cite{SG_Modeling_ref3,D2D_underlaid_ref1} such that
the benefits of proximity transmissions can be fully exploited. In
terms of the overlaid mode, orthogonal resources are partially allocated
to cellular users and D2D users, respectively \cite{D2D_overlaid_ref1,D2D_overlaid_ref2}.
Since only co-tier interference exists, the performance degradation
caused by interference can be more readily handled. Nevertheless,
it is crucial to wisely implement resource allocation in order to
better utilize the available resources. In \cite{D2D_overlaid_ref2},
authors show that, even in the overlaid case, D2D communication can
significantly improve the per-user average rate through carefully
tuning frequency allocation for D2D pairs.

In both coexisting modes, efficient multiple access methods should
be applied in order to support more cellular users and D2D users.
One of the most prevalent methods is orthogonal frequency-division
multiple access (OFDMA) \cite{OFDMA_ref}. Using OFDMA, orthogonal
subcarriers or OFDMA tones%
\footnote{In the following, we use OFDMA tones to replace OFDMA orthogonal resources.%
} are allocated to individual users so as to mitigate the interference
over one tone. In order to improve spectrum reuse, D2D users can be
underlaid with cellular network, i.e., OFDMA tones are reused by cellular
users and D2D users. Accordingly, cross-tier interference is introduced
between cellular users and D2D users, especially when users are densely
distributed. To avoid the mutual interference and improve the quality
of service (QoS) of individual users, it is preferable to use the
overlaid mode, where OFDMA tones are separately allocated to cellular
users and D2D users. However, due to the scarcity of spectrum resources,
limited OFDMA tones are available such that the number of users that
can be simultaneously served is critically restricted. Therefore,
it is challenging to admit more users over the limited spectrum resources.
Recently, sparse code multiple access (SCMA) has been proposed for
5G wireless networks, which has the potential to enable massive connectivity
by allowing non-orthogonal spectrum resource sharing \cite{SCMA_original}.
Since users occupying the same OFDMA resources can be distinguished
using different SCMA constellations, more orthogonal resources, i.e.,
codebooks, can be provided by SCMA compared to OFDMA. As a result,
overloading%
\footnote{Overloading can be achieved when the number of available SCMA layers
is greater than the number of OFDMA tones. For instance, using typical
SCMA setting, 6 SCMA layers can be multiplexed over 4 OFDMA tones.
Therefore, system can be overloaded by admitting more users. As will
be discussed later, each SCMA layer is assigned with one SCMA codebook.%
} gain can be yielded by SCMA. Besides, depending on the resource
allocation rule, the number of users occupying the same codebooks
will be decreased as more orthogonal resources are to be shared. Hence,
the interference over one codebook could be mitigated, thereby boosting
the performance of individual users. For these reasons, it is intuitively
more suitable for SCMA to be applied in the D2D and cellular hybrid
network, especially when communication devices are densely deployed.
However, since SCMA codewords are spread to multiple OFDMA tones \cite{SCMA_codebook_design},
great difficulty has been encountered in analytically characterizing
the interference statistics of cellular users and D2D users. Thus,
the benefits of SCMA in enhancing the hybrid network performance remain
to be explored. Worsestill, as overloading gain can only be achieved
when codebooks are uniquely occupied by individual users, it is doubtful
whether the overloading gain can also be achieved when codebook reuse
is enabled.

Motivated by the above discussion, we consider a D2D and cellular
hybrid network, where SCMA is employed by both uplink cellular transmissions
and D2D transmissions. Using stochastic geometry, we present a tractable
model to evaluate the performance of cellular network and D2D network
by characterizing the signal-to-interference ratio (SIR) statistics
of a typical cellular uplink and a typical D2D link. Moreover, the
potential of SCMA in supporting more active users and enhancing spectrum
utilization has been investigated using area spectral efficiency (ASE).
The main contributions and outcomes of this paper are summarized as
follows:
\begin{itemize}
\item \textbf{Analytical Framework:} To our best knowledge, no previous
work has investigated the performance of SCMA in large scale networks.
To fill in the gap, we have developed an analytical framework for
the design and analysis of SCMA enhanced D2D and cellular hybrid network.
The framework could capture the impact of important parameters on
the performance of the hybrid network, including user spatial distributions,
wireless channel propagation model, D2D mode selection rule and SCMA
parameters, etc.
\item \textbf{SCMA VS OFDMA:} Based on the analytical framework, closed-form
expressions of approximate system ASE are derived in the D2D underlaid/overlaid
cellular network. Furthermore, we compare the performance of SCMA
with OFDMA, in the underlaid mode. In particular, we quantify the
gain of SCMA over OFDMA, termed ASE gain. Through the comparison of
ASE gain with overloading gain, it is numerically shown that the overloading
gain can be almost achieved even when codebook reuse is allowed.
\item \textbf{System Design Guidance:} In the underlaid mode, we enable
D2D transmitters to use an activated probability to control the generated
cross-tier interference to cellular users. The optimal activated probability
is obtained such that proportional fairness utility function is maximized.
In the overlaid mode, we study SCMA codebook allocation for cellular
network and D2D network. Specifically, we find out the optimal codebook
allocation rule when cellular users are densely deployed, in order
to maximize the proportional fairness utility function. The results
can serve as a guideline for the efficient design of SCMA mechanisms
and resource allocation in the D2D and cellular hybrid network.
\end{itemize}

The remainder of this paper is organized as follows. We first present
the system model and performance metrics in Section \ref{sec:System-Model}.
Performance analysis of D2D underlaid cellular network is then presented
in Section \ref{sec:Underlaid D2D}, where comparison between OFDMA
and SCMA is made. In Section \ref{sec:Overlaid D2D}, we evaluate
the performance of SCMA in the D2D overlaid cellular network. Numerical
results are given in Section \ref{sec:Numerical Results} and conclusions
are provided in Section \ref{sec:Conclusion}.

\section{System Model\label{sec:System-Model}}

\begin{figure}[t]
\begin{centering}
\includegraphics[width=3in]{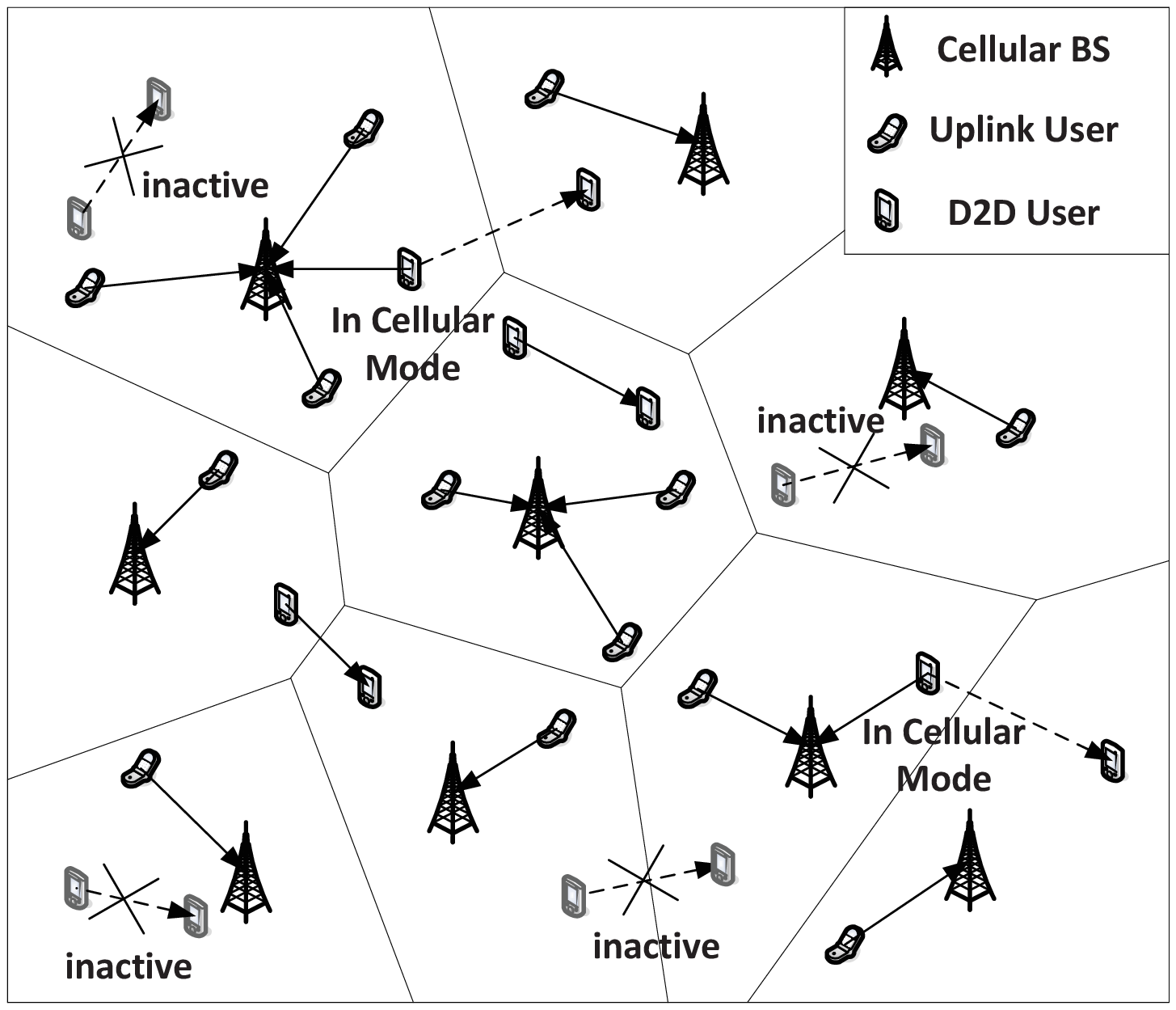}
\par\end{centering}

\caption{\label{fig:scenario}Illustration of a D2D hybrid cellular network.
D2D users select cellular mode for data transmission when potential
D2D link length is too large. Meanwhile, D2D links are activated with
a probability $q_{\mathrm{D}}$ to control interference to cellular
users in the underlaid mode.}

\end{figure}

\subsection{Network Model}

We consider a D2D and cellular hybrid network (see Fig. \ref{fig:scenario}),
consisting of cellular uplinks and unidirectional D2D links. Stochastic
geometry model is used to characterize the locations of BSs and users
in the hybrid network \cite{SG_Modeling_ref1,SG_Modeling_ref2,SG_Modeling_ref3}.
In particular, cellular BSs are assumed to be spatially distributed
over the infinite two-dimensional plane $\mathbb{R}^{2}$, according
to a homogeneous Poisson Point Process (HPPP) $\Pi_{\mathrm{BS}}=\left\{ \mathrm{BS}_{i}\right\} $
$\left(i\in\mathbb{N}\right)$ with intensities $\lambda_{\mathrm{BS}}$.
The locations of uplink cellular users and D2D transmitters are also
modeled using HPPPs $\Pi_{\mathrm{U}}=\left\{ \mathrm{U}_{j}\right\} $
and $\Pi_{\mathrm{D}}=\left\{ \mathrm{DT}_{k}\right\} $ $\left(j,k\in\mathbb{N}\right)$
with intensity $\lambda_{\mathrm{U}}$ and $\lambda_{\mathrm{D}}$,
respectively, and constant transmit power $P_{\mathrm{U}}$ and $P_{\mathrm{D}}$,
respectively. The cellular user $\mathrm{U}_{j}$ is associated with
the geographically closest BS, the distance between which is denoted
by $r_{\mathrm{U},j}$. The D2D transmitter $\mathrm{DT}_{k}$ connects
to an intended D2D receiver $\mathrm{DR}_{k}$ with isotropic direction
at distance $r_{\mathrm{D},k}$ away. Similarly with \cite{D2D_link_ref_1,Spectrum_sharing_ref},
we assume that $r_{\mathrm{D},k}$ follows the Rayleigh distribution
with probability density function (PDF) given by
\begin{equation}
f_{r_{\mathrm{D},k}}\left(x\right)=2\pi\xi x\exp\left(-\xi\pi x^{2}\right),\: x\geq0.\label{eq:D2D link length distribution}
\end{equation}
According to (\ref{eq:D2D link length distribution}), $r_{\mathrm{D},k}$
can be very large, which, however, makes D2D transmission lose its
proximity merit. Therefore, for practical concerns, we introduce a
distance based mode selection rule for D2D users \cite{cellular_D2D_mode_selection,Spectrum_sharing_ref}.
Specifically, $\mathrm{DT}_{k}$ is allowed to use D2D mode only when
$r_{\mathrm{D},k}$ is smaller than a distance threshold $\tau_{\mathrm{dis}}$.
Otherwise, $\mathrm{DT}_{k}$ uses cellular mode with transmit power
$P_{\mathrm{U}}$ by connecting to the nearest BS. Thus, $\tau_{\mathrm{dis}}$
serves as a tunable parameter to control the traffic load from D2D
network to cellular network. According to this rule, a new Poisson
Point Process (PPP) $\Pi_{\mathrm{UT}}=\left\{ \mathrm{U}_{j}\right\} $
is formed by uplink cellular users and cellular mode D2D transmitters
with intensity $\lambda_{\mathrm{UT}}=\lambda_{\mathrm{U}}+\lambda_{\mathrm{D}}\mathbb{P}\left(r_{\mathrm{D},k}>\tau_{\mathrm{dis}}\right)$.
Meanwhile, transmitters selecting D2D mode form a PPP $\Pi_{\mathrm{DT}}=\left\{ \mathrm{DT}_{k}\right\} $
with intensity $\lambda_{\mathrm{DT}}=\lambda_{\mathrm{D}}\mathbb{P}\left(r_{\mathrm{D},k}\leq\tau_{\mathrm{dis}}\right)$.
Besides, we consider that uplink cellular users and D2D transmitters
always have data to transmit.

In terms of the coexistence of cellular network and D2D network, underlaid
mode and overlaid mode are taken into account. In the underlaid mode,
all available resources are universally reused by cellular users and
D2D users. Consequently, cross-tier interference is introduced between
cellular users and D2D users. In contrast, resources are partially
allocated to cellular users and D2D users, respectively, such that
only co-tier interference exists.

We evaluate the performance of a typical cellular uplink and a typical
D2D link. $\mathrm{U}_{0}$ ($\mathrm{DT}_{0}$) and $\mathrm{BS}_{0}$
($\mathrm{DR}_{0}$) are the associated transmitter and receiver of
the typical cellular (D2D) link, respectively. According to Slivnyak's
Theorem \cite{book_stochastic_geometry}, the performance of typical
links can be used to evaluate the performance of the other links.

\subsection{Channel Model\label{sub:Channel-Model}}

Consider that channel gain consists of a path loss component with
the path loss exponent $\alpha$ $\left(\alpha>2\right)$ and a distance-independent
small-scale fading component. Besides, we use independently and identically
distributed (i.i.d.) Rayleigh fading, i.e., $h\sim\mathcal{CN}\left(0,1\right)$,
to model the small-scale fading.

\subsection{Multiple Access Schemes\label{sub:Multiple Access Schemes}}

We consider OFDMA and SCMA as two multiple access schemes for cellular
users and D2D users in the hybrid network.  OFDMA is a typical and
most widely used multiple access scheme in wireless networks, especially
for the uplink transmission in cellular networks. Moreover, SCMA is
designed based on OFDMA. Therefore, it is straightforward and convincing
to take OFDMA as a benchmark to demonstrate the advantage of SCMA.
In order to facilitate the comparison, we only consider the underlaid
case when considering OFDMA. In the following, we describe how to
implement the two schemes.

\textbf{1) OFDMA}. For cellular transmission, each cellular user in
one cell is randomly and independently assigned with one of $K$ OFDMA
tones, which are orthogonal in time and frequency. Meanwhile, OFDMA
tones can be reused by BSs of different cells. In consequence, only
inter-cell interference exists, while no intra-cell interference exists.
According to this allocation rule, when the number of cellular users
within a cell is larger than that of OFDMA tones, randomly selected
$K$ cellular users are activated and the remaining users are kept
inactive due to the unavailability of OFDMA resources. For D2D transmission,
each D2D pair is also randomly and independently allocated with one
OFDMA tone. Since the same OFDMA tones are reused by cellular users
and D2D users, mutual interference is introduced.

We denote $x_{\mathrm{U}_{i}}^{\mathrm{O}}$ and $x_{\mathrm{DT}_{j}}^{\mathrm{O}}$
as the data symbols sent from $\mathrm{U}_{i}$ and $\mathrm{DT}_{j}$,
respectively. Accordingly, the received signal at $\mathrm{BS}_{0}$
can be expressed as
\begin{equation}
y_{\mathrm{BS}_{0}}^{\mathrm{O}}=\sqrt{P_{\mathrm{U}}}h_{\mathrm{BS}_{0},\mathrm{U}_{0}}x_{\mathrm{U}_{0}}^{\mathrm{O}}+\underset{\tiny{\mathrm{U}_{i}\in\tilde{\Pi}_{\mathrm{UT}}^{\mathrm{O}}}}{\sum}\sqrt{P_{\mathrm{U}}}h_{\mathrm{BS}_{0},\mathrm{U}_{i}}x_{\mathrm{U}_{i}}^{\mathrm{O}}+\underset{\tiny{\mathrm{DT}_{j}\in\Pi_{\mathrm{DT}}^{\mathrm{O}}}}{\sum}\sqrt{P_{\mathrm{D}}}h_{\mathrm{BS}_{0},\mathrm{DT}_{j}}x_{\mathrm{DT}_{j}}^{\mathrm{O}}+n_{0},\label{eq:recv signal at BS OFDMA}
\end{equation}
where $\Pi_{\mathrm{DT}}^{\mathrm{O}}$ denotes the set of active
D2D transmitters using the same OFDMA tone with $\mathrm{U}_{0}$.
$\tilde{\Pi}_{\mathrm{UT}}^{\mathrm{O}}=\Pi_{\mathrm{UT}}^{\mathrm{O}}\backslash\left\{ \mathrm{U}_{0}\right\} $,
where $\Pi_{\mathrm{UT}}^{\mathrm{O}}$ is the set of active cellular
users using the same OFDMA tone with $\mathrm{U}_{0}$. In the following,
$h_{\mathrm{Y},\mathrm{X}}$ denotes the channel from $\mathrm{X}$
to $\mathrm{Y}$. $n_{0}$ is the additive Gaussian noise.

Similarly, the received signal at $\mathrm{DR}_{0}$ using OFDMA can
be expressed as
\begin{equation}
y_{\mathrm{DR}_{0}}^{\mathrm{O}}=\sqrt{P_{\mathrm{D}}}h_{\mathrm{DR}_{0},\mathrm{DT}_{0}}x_{\mathrm{DT}_{0}}^{\mathrm{O}}+\underset{\tiny{\mathrm{DT}_{j}\in\tilde{\Pi}_{\mathrm{DT}}^{\mathrm{O}}}}{\sum}\sqrt{P_{\mathrm{D}}}h_{\mathrm{DR}_{0},\mathrm{DT}_{j}}x_{\mathrm{DT}_{j}}^{\mathrm{O}}+\underset{\tiny{\mathrm{U}_{i}\in\Pi_{\mathrm{UT}}^{\mathrm{O}}}}{\sum}\sqrt{P_{\mathrm{U}}}h_{\mathrm{DR}_{0},\mathrm{U}_{i}}x_{\mathrm{U}_{i}}^{\mathrm{O}}+n_{0},\label{eq:recv signal at D2D OFDMA}
\end{equation}
where $\tilde{\Pi}_{\mathrm{DT}}^{\mathrm{O}}=\Pi_{\mathrm{DT}}^{\mathrm{O}}\backslash\left\{ \mathrm{DT}_{0}\right\} $.

\begin{figure}[t]
\begin{centering}
\includegraphics[width=3.5in]{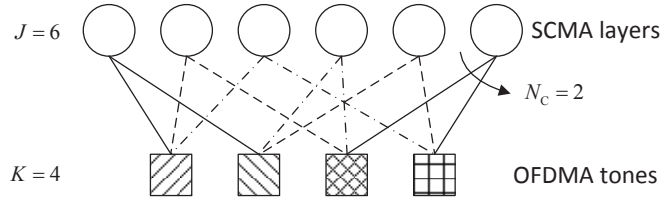}
\par\end{centering}

\caption{\label{fig:Mapping SCMA OFDMA}An exemplary mapping between SCMA layers
and OFDMA tones.}

\end{figure}

\textbf{2) SCMA}. An SCMA encoder is defined as a function, which
maps a binary stream of $\log_{2}\left(M\right)$ bits to a $K$-dimensional
complex codebook of size $M$ \cite{SCMA_original}. The $K$-dimensional
complex codewords of each codebook are sparse vectors with $N_{\mathrm{C}}$
$\left(2\leq N_{\mathrm{C}}<K\right)$ non-zero entries. Using SCMA,
each SCMA layer is assigned with a codebook and users' data is transmitted
using the codeword from the codebook over $K$ OFDMA tones. An exemplary
mapping relationship between SCMA layers and OFDMA tones is illustrated
in Fig. \ref{fig:Mapping SCMA OFDMA}. Therefore, multiple access
is achieved by sharing the same time-frequency resources among the
SCMA layers. According to \cite{SCMA_original}, the maximum number
of codebooks $J$ is a function of the codeword length $K$ and the
number of non-zero elements in the codeword $N_{\mathrm{C}}$. The
generation of codebooks is equivalent to selecting $N_{\mathrm{C}}$
positions out of $K$ elements. Hence, $J=\mathrm{C}_{K}^{N_{\mathrm{C}}}$.
As a result, the received signal after layer multiplexing can be obtained
as $\mathbf{y}=\sum_{j=1}^{J}\sqrt{P_{\mathrm{U}}}\mathrm{diag}\left(\mathbf{h}_{j}\right)\mathbf{x}_{j}+\mathbf{n}_{0},$
where $\mathbf{x}_{j}$ is the $K\times1$ SCMA codeword of the layer
$j$, $\mathbf{h}_{j}$ is the $K\times1$ channel vector of layer
$j$ and $\mathbf{n}_{0}$ is the vector of additive Gaussian noise.
$\mathrm{diag}\left(\cdot\right)$ is defined as the operation, which
turns a vector into a matrix by putting the vector on the main diagonal
of the matrix.

Due to the sparsity of SCMA codewords, multi-user detection based
on message passing algorithm (MPA) can be implemented at the SCMA
receiver with low complexity \cite{MPA_ref}. Using ideal MPA receiver,
codewords from different layers can be decoded without interfering
with each other. Therefore, the codebooks allocated to different layers
can be considered as orthogonal resources. Consequently, interference
occurs only when the same layer or equivalently the same codebook
is reused by more than one user. Since $J$ layers can be multiplexed
over $K$ resources, we further define the SCMA overloading factor
as
\begin{equation}
\eta_{\mathrm{overload}}=\frac{J}{K}.\label{eq:overloading factor}
\end{equation}

For cellular transmission, each BS randomly and independently selects
one of the available codebooks to serve one connected cellular user.
Similarly, if the number of the connected cellular users at a BS is
larger than that of the available codebooks, randomly selected $J$
cellular users are kept active. For D2D transmission, each D2D pair
is randomly and independently allocated with one codebook to transmit
with. 

We consider two coexisting modes, i.e., underlaid mode and overlaid
mode. Let $x_{\mathrm{U}_{i},m}^{\mathrm{S}}$ and $x_{\mathrm{DT}_{i},m}^{\mathrm{S}}$
denote the SCMA codewords sent by $\mathrm{U}_{i}$ and $\mathrm{DT}_{j}$,
respectively, over the $m$th OFDMA tone. Then, the received signal
at $\mathrm{BS}_{0}$ is given by
\begin{equation}
y_{\mathrm{BS}_{0}}^{\mathrm{S}}=\sqrt{P_{\mathrm{U}}^{\dagger}}g_{\mathrm{BS}_{0},\mathrm{U}_{0}}x_{\mathrm{U}_{0},m}^{\mathrm{S}}+\underset{\tiny{\mathrm{U}_{i}\in\tilde{\Pi}_{\mathrm{UT}}^{\mathrm{S}}}}{\sum}\sqrt{P_{\mathrm{U}}^{\dagger}}g_{\mathrm{BS}_{0},\mathrm{U}_{i}}x_{\mathrm{U}_{i},m}^{\mathrm{S}}+\mathbf{1}_{\mathrm{U}}\underset{\tiny{\mathrm{DT}_{j}\in\Pi_{\mathrm{DT}}^{\mathrm{S}}}}{\sum}\sqrt{P_{\mathrm{D}}^{\dagger}}g_{\mathrm{BS}_{0},\mathrm{DT}_{j}}x_{\mathrm{DT}_{j},m}^{\mathrm{S}}+n_{0},\label{eq:recv signal at BS SCMA}
\end{equation}
where $g_{\mathrm{BS}_{0},\mathrm{U}_{i}}=\overset{N_{\mathrm{C}}}{\underset{m=1}{\sum}}h_{\mathrm{BS}_{0},\mathrm{U}_{i},m}$,
$g_{\mathrm{BS}_{0},\mathrm{DT}_{j}}=\overset{N_{\mathrm{C}}}{\underset{m=1}{\sum}}h_{\mathrm{BS}_{0},\mathrm{DT}_{j},m}$
and $h_{\mathrm{Y},\mathrm{X},m}$ denotes the channel from $\mathrm{X}$
to $\mathrm{Y}$ over the $m$th OFDMA tone. $\mathbf{1}_{\mathrm{U}}$
is the indicator function, which equals 1 in the underlaid mode and
0 in the overlaid mode. $P_{\mathrm{U}}^{\dagger}=\frac{P_{\mathrm{U}}}{N_{\mathrm{C}}}$
and $P_{\mathrm{D}}^{\dagger}=\frac{P_{\mathrm{D}}}{N_{\mathrm{C}}}$,
since transmit power is assumed to be uniformly allocated over $N_{\mathrm{C}}$
OFDMA tones. $\Pi_{\mathrm{DT}}^{\mathrm{S}}$ denotes the set of
active D2D transmitters, which use the same codebook with $\mathrm{BS}_{0}$.
$\tilde{\Pi}_{\mathrm{UT}}^{\mathrm{S}}=\Pi_{\mathrm{UT}}^{\mathrm{S}}\backslash\left\{ \mathrm{U}_{0}\right\} $,
where $\Pi_{\mathrm{UT}}^{\mathrm{S}}$ denotes the set of active
cellular users using the same codebook with $\mathrm{U}_{0}$.

Similarly, the received signal at $\mathrm{DR}_{0}$ can be expressed
as
\begin{equation}
y_{\mathrm{DR}_{0}}^{\mathrm{S}}=\sqrt{P_{\mathrm{D}}^{\dagger}}g_{\mathrm{DR}_{0},\mathrm{DT}_{0}}x_{\mathrm{DT}_{0},m}^{\mathrm{S}}+\underset{\tiny{\mathrm{DT}_{j}\in\tilde{\Pi}_{\mathrm{DT}}^{\mathrm{S}}}}{\sum}\sqrt{P_{\mathrm{D}}^{\dagger}}g_{\mathrm{DR}_{0},\mathrm{DT}_{j}}x_{\mathrm{DT}_{j},m}^{\mathrm{S}}+\mathbf{1}_{\mathrm{U}}\underset{\tiny{\mathrm{U}_{i}\in\Pi_{\mathrm{UT}}^{\mathrm{S}}}}{\sum}\sqrt{P_{\mathrm{U}}^{\dagger}}g_{\mathrm{DR}_{0},\mathrm{U}_{i}}x_{\mathrm{U}_{i},m}^{\mathrm{S}}+n_{0},\label{eq:recv signal at D2D SCMA}
\end{equation}
where $g_{\mathrm{DR}_{0},\mathrm{DT}_{j}}=\overset{N_{\mathrm{C}}}{\underset{m=1}{\sum}}h_{\mathrm{DR}_{0},\mathrm{DT}_{j},m}$,
$g_{\mathrm{DR}_{0},\mathrm{U}_{i}}=\overset{N_{\mathrm{C}}}{\underset{m=1}{\sum}}h_{\mathrm{DR}_{0},\mathrm{U}_{i},m}$
and $\tilde{\Pi}_{\mathrm{DT}}^{\mathrm{S}}=\Pi_{\mathrm{\mathrm{DT}}}^{\mathrm{S}}\backslash\left\{ \mathrm{DT}_{0}\right\} $.

\subsection{Performance Metrics\label{sub:Performance Metrics}}

We use ASE $\left[\mathrm{bits}/\left(\mathrm{s\cdot Hz\cdot m^{2}}\right)\right]$
to evaluate the performance of the D2D and cellular hybrid network
\cite{ASE_1}. Specifically, ASE of the cellular network is defined
as 
\begin{equation}
\mathcal{A}_{\mathrm{C}}=q_{\mathrm{U}}\lambda_{\mathrm{UT}}\mathbb{P}\left\{ \mathrm{SIR}_{\mathrm{BS}_{0}}>\tau_{\mathrm{BS}}\right\} \mathrm{log}\left(1+\tau_{\mathrm{BS}}\right),\label{eq:cellular ASE}
\end{equation}
where $q_{\mathrm{U}}$ is the access probability of cellular users,
$\mathrm{SIR}_{\mathrm{BS}_{0}}$ is the SIR at $\mathrm{BS}_{0}$
and $\tau_{\mathrm{BS}}$ is the corresponding SIR threshold. Likewise,
ASE of the D2D network is defined as
\begin{equation}
\mathcal{A}_{\mathrm{D}}=q_{\mathrm{D}}\lambda_{\mathrm{DT}}\mathbb{P}\left\{ \mathrm{SIR}_{\mathrm{DR}_{0}}>\tau_{\mathrm{DR}}\right\} \mathrm{log}\left(1+\tau_{\mathrm{DR}}\right),\label{eq:D2D rate}
\end{equation}
where $q_{\mathrm{D}}$ is the activated probability of D2D transmitters,
$\mathrm{SIR}_{\mathrm{DR}_{0}}$ is the SIR at $\mathrm{DR}_{0}$
and $\tau_{\mathrm{DR}}$ is the corresponding SIR threshold. Note
that the effect of noise is ignored in the interference-limited network.

In the following, $F\left(\cdot\right)$ and $f\left(\cdot\right)$
denote the cumulative distribution function (CDF) and PDF, respectively.
Meanwhile, the notations used throughout this paper are summarized
in Table \ref{tab:Summary-of-Notation}.

\begin{table}[t]
\centering{}\caption{\label{tab:Summary-of-Notation}Summary of Notations}
\begin{tabular}{|c|c|c|c|}
\hline 
Symbol & Meaning & Symbol & Meaning\tabularnewline
\hline 
\hline 
$\mathrm{BS}_{i}$, $\mathrm{U}_{j}$ & $i$th BS, $j$th uplink user & \multirow{2}{*}{$y_{\mathrm{DR}_{0}}^{\mathrm{O}}$, $y_{\mathrm{DR}_{0}}^{\mathrm{S}}$} & received signals at $\mathrm{DR}_{0}$ \tabularnewline
\cline{1-2} 
$\mathrm{DT}_{m}$ & $m$th D2D transmitter &  & using OFDMA and SCMA\tabularnewline
\hline 
$\mathrm{DR}_{n}$ & $n$th D2D receiver & $P_{\mathrm{U}}^{\dagger}$ & $P_{\mathrm{U}}^{\dagger}=\frac{P_{\mathrm{U}}}{N_{\mathrm{C}}}$\tabularnewline
\hline 
$\Pi_{\mathrm{BS}}$ & $\Pi_{\mathrm{BS}}=\left\{ \mathrm{BS}_{i}\right\} $ & \multicolumn{1}{c|}{$P_{\mathrm{D}}^{\dagger}$} & $P_{\mathrm{D}}^{\dagger}=\frac{P_{\mathrm{D}}}{N_{\mathrm{C}}}$\tabularnewline
\hline 
$\Pi_{\mathrm{U}}$ &  uplink cellular users set & $\eta_{\mathrm{P}}$ & $\eta_{\mathrm{P}}=\frac{P_{\mathrm{D}}}{P_{\mathrm{U}}}$\tabularnewline
\hline 
\multirow{2}{*}{$\Pi_{\mathrm{UT}}$} & set of $\Pi_{\mathrm{U}}$ and D2D & \multirow{3}{*}{$\Pi_{\mathrm{UT}}^{\mathrm{O}}$, $\Pi_{\mathrm{DT}}^{\mathrm{O}}$} & sets of active cellular users\tabularnewline
 & users in cellular mode &  & and D2D transmitters\tabularnewline
\cline{1-2} 
$\Pi_{\mathrm{D}}$ & D2D transmitters set &  & over OFDMA tone\tabularnewline
\hline 
\multirow{2}{*}{$\Pi_{\mathrm{DT}}$} & set of D2D users & \multicolumn{1}{c|}{$\tilde{\Pi}_{\mathrm{UT}}^{\mathrm{O}}$} & $\tilde{\Pi}_{\mathrm{UT}}^{\mathrm{O}}=\Pi_{\mathrm{UT}}^{\mathrm{O}}\backslash\left\{ \mathrm{U}_{0}\right\} $\tabularnewline
\cline{3-4} 
 & in D2D mode & $\tilde{\Pi}_{\mathrm{DT}}^{\mathrm{O}}$ & $\tilde{\Pi}_{\mathrm{DT}}^{\mathrm{O}}=\Pi_{\mathrm{DT}}^{\mathrm{O}}\backslash\left\{ \mathrm{DT}_{0}\right\} $\tabularnewline
\hline 
$\lambda_{\mathrm{BS}}$ & intensity of $\Pi_{\mathrm{BS}}$ & \multirow{3}{*}{$\Pi_{\mathrm{UT}}^{\mathrm{S}}$, $\Pi_{\mathrm{DT}}^{\mathrm{S}}$} & set of active cellular users\tabularnewline
\cline{1-2} 
$\lambda_{\mathrm{U}}$, $\lambda_{\mathrm{UT}}$ & intensities of $\Pi_{\mathrm{U}}$ and $\Pi_{\mathrm{UT}}$ &  & and D2D transmitters\tabularnewline
\cline{1-2} 
$\lambda_{\mathrm{D}}$, $\lambda_{\mathrm{DT}}$ & intensities of $\Pi_{\mathrm{D}}$ and $\Pi_{\mathrm{DT}}$ &  & using one SCMA codebook\tabularnewline
\hline 
$P_{\mathrm{U}}$, $P_{\mathrm{D}}$ & transmit power of $\mathrm{U}_{i}$ and $\mathrm{DT}_{m}$ & $\tilde{\Pi}_{\mathrm{UT}}^{\mathrm{S}}$ & $\tilde{\Pi}_{\mathrm{UT}}^{\mathrm{S}}=\Pi_{\mathrm{UT}}^{\mathrm{S}}\backslash\left\{ \mathrm{U}_{0}\right\} $\tabularnewline
\hline 
$r_{\mathrm{U},j}$ & distance from $\mathrm{U}_{j}$ to $\mathrm{BS}_{j}$ & $\tilde{\Pi}_{\mathrm{DT}}^{\mathrm{S}}$ & $\tilde{\Pi}_{\mathrm{DT}}^{\mathrm{S}}=\Pi_{\mathrm{DT}}^{\mathrm{S}}\backslash\left\{ \mathrm{DT}_{0}\right\} $\tabularnewline
\hline 
$r_{\mathrm{D},m}$ & distance from $\mathrm{DT}_{m}$ to $\mathrm{DR}_{m}$ & \multirow{2}{*}{$\tau_{\mathrm{BS}}$, $\tau_{\mathrm{DR}}$} & SIR thresholds at BSs \tabularnewline
\cline{1-2} 
$\tau_{\mathrm{dis}}$ & mode selection threshold &  & and D2D receivers\tabularnewline
\hline 
$\xi$ & D2D link length parameter & $\tilde{\tau}_{\mathrm{BS}}$ & $\tilde{\tau}_{\mathrm{BS}}=\frac{\tau_{\mathrm{BS}}}{N_{\mathrm{C}}}$\tabularnewline
\hline 
$\alpha$, $\delta$ & path loss exponent, $\delta=\frac{2}{\alpha}$ & $\tilde{\tau}_{\mathrm{DR}}$ & $\tilde{\tau}_{\mathrm{DR}}=\frac{\tau_{\mathrm{DR}}}{N_{\mathrm{C}}}$\tabularnewline
\hline 
$h_{\mathrm{Y},\mathrm{X}}$ & Rayleigh fading from $\mathrm{X}$ to $\mathrm{Y}$ & $\eta_{\mathrm{overload}}$ & overloading factor\tabularnewline
\hline 
$h_{\mathrm{Y},\mathrm{X},m}$ & $h_{\mathrm{Y},\mathrm{X}}$ at the $m$th OFDMA tone & $\eta_{\mathrm{ASE}}$ & ASE gain\tabularnewline
\hline 
$K$ & number of OFDMA tones & $\hat{\eta}_{\mathrm{ASE}}$ & ASE gain when $\tau_{\mathrm{BS}}=\tau_{\mathrm{DR}}$\tabularnewline
\hline 
$M$ & codebook size & $q_{\mathrm{U}}$ & cellular access probability\tabularnewline
\hline 
\multirow{2}{*}{$N_{\mathrm{C}}$} & number of non-zero elements  & \multicolumn{1}{c|}{$q_{\mathrm{D}}$} & D2D activated probability\tabularnewline
\cline{3-4} 
 & in SCMA codewords & \multirow{3}{*}{$q_{\mathrm{U}}^{\mathrm{O}}$, $q_{\mathrm{U}}^{\mathrm{S}}$} & access probabilities of\tabularnewline
\cline{1-2} 
\multirow{2}{*}{$y_{\mathrm{BS}_{0}}^{\mathrm{O}}$, $y_{\mathrm{BS}_{0}}^{\mathrm{S}}$} & received signals at $\mathrm{BS}_{0}$  &  & cellular users using \tabularnewline
 & using OFDMA and SCMA &  & OFDMA and SCMA\tabularnewline
\hline 
\end{tabular}
\end{table}

\section{D2D Underlaid Cellular Network\label{sec:Underlaid D2D}}

In this section, we first compare the performance of OFDMA and SCMA
in the underlaid mode through ASE. Afterward, we endow D2D users with
an activated probability to balance the cross-tier interference and
search for the optimal activated probability of D2D users to maximize
the utility function defined based on proportional fairness.

\subsection{OFDMA VS SCMA\label{sub:OFDMA-VS-SCMA}}

In this part, we investigate the performance of OFDMA and SCMA. To
this end, we first determine the access probability of cellular users,
which is defined as the probability that a BS uses available orthogonal
resources (OFDMA tones or SCMA codebooks) to serve its associated
cellular users, according to the following lemma.

\begin{lemma}

\noindent The probability that a BS allocates one orthogonal resource
to one of the associated cellular users is given by
\begin{equation}
q_{\mathrm{U}}=1-\sum_{m=N_{\mathrm{R}}+1}^{\infty}\frac{m-1}{m}\mathbb{P}\left\{ N_{\mathrm{U}}=m\right\} ,\label{eq:cellular access probability}
\end{equation}
where $N_{\mathrm{R}}$ is the number of available resources and $N_{\mathrm{U}}$
is the number of cellular users served by the BS.

\label{lemma: access probability OFDMA}

\end{lemma}

\begin{proof}See Appendix A-I in \cite{Cognitive_D2D_cellular}.\end{proof}

According to \cite{App_using_PPP_1}, the probability mass function
(PMF) of $N_{\mathrm{U}}$ can be obtained as $\mathbb{P}\left\{ N_{\mathrm{U}}=m\right\} =\frac{b^{b}\Gamma(m+b)}{\Gamma\left(b\right)\Gamma(m+1)}\frac{\left(\mathbb{E}\left[N_{\mathrm{U}}\right]\right)^{m}}{\left(b+\mathbb{E}\left[N_{\mathrm{U}}\right]\right)^{m+b}}$,
where $b=3.575$, $\Gamma\left(\cdot\right)$ is the standard gamma
function and $\mathbb{E}\left[N_{\mathrm{U}}\right]=\frac{\lambda_{\mathrm{UT}}}{\lambda_{\mathrm{BS}}}$
is the average number of cellular users connected to one BS.

We substitute $N_{\mathrm{R}}$ with $K$ and $J$, respectively,
in Lemma \ref{lemma: access probability OFDMA} to obtain the access
probabilities $q_{\mathrm{U}}^{\mathrm{O}}$ and $q_{\mathrm{U}}^{\mathrm{S}}$
when OFDMA and SCMA are used. Then, the densities of active cellular
users using the same OFDMA tone and the same SCMA codebook are $\frac{q_{\mathrm{U}}^{\mathrm{O}}\lambda_{\mathrm{UT}}}{K}$
and $\frac{q_{\mathrm{U}}^{\mathrm{S}}\lambda_{\mathrm{UT}}}{J}$,
respectively.

Next, we analyze the ASE of the underlaid scenario when OFDMA is applied.
According to the definitions of ASE in Section \ref{sub:Performance Metrics},
the coverage probabilities $\mathrm{CP}_{\mathrm{BS}}^{\mathrm{OU}}=\mathbb{P}\left\{ \mathrm{SIR}_{\mathrm{BS}_{0}}^{\mathrm{OU}}>\tau_{\mathrm{BS}}\right\} $
and $\mathrm{CP}_{\mathrm{DR}}^{\mathrm{OU}}=\mathbb{P}\left\{ \mathrm{SIR}_{\mathrm{DR}_{0}}^{\mathrm{OU}}>\tau_{\mathrm{DR}}\right\} $
should be calculated in order to derive ASE. Based on (\ref{eq:recv signal at BS OFDMA}),
the SIR at $\mathrm{BS}_{0}$ using OFDMA in the underlaid mode can
be expressed as
\begin{equation}
\mathrm{SIR}_{\mathrm{BS}_{0}}^{\mathrm{OU}}=\frac{P_{\mathrm{U}}r_{\mathrm{U},0}^{-\alpha}\left\Vert h_{\mathrm{BS}_{0},\mathrm{U}_{0}}\right\Vert ^{2}}{I_{\mathrm{C},\mathrm{C}}^{\mathrm{O}}+I_{\mathrm{C},\mathrm{D}}^{\mathrm{O}}},\label{eq:SIR at BS OFDMA}
\end{equation}
where $I_{\mathrm{C},\mathrm{C}}^{\mathrm{O}}=\underset{\tiny{\mathrm{U}_{i}\in\tilde{\Pi}_{\mathrm{UT}}^{\mathrm{O}}}}{\sum}P_{\mathrm{U}}\left\Vert h_{\mathrm{BS}_{0},\mathrm{U}_{i}}\right\Vert ^{2}\left\Vert \mathrm{U}_{i}-\mathrm{BS}_{0}\right\Vert ^{-\alpha}$
is the inter-cell interference from active cellular users using the
same OFDMA tone with $\mathrm{U}_{0}$ outside the coverage of $\mathrm{BS}_{0}$
and $I_{\mathrm{C},\mathrm{D}}^{\mathrm{O}}=\underset{\tiny{\mathrm{DT}_{j}\in\Pi_{\mathrm{DT}}^{\mathrm{O}}}}{\sum}P_{\mathrm{D}}\left\Vert h_{\mathrm{BS}_{0},\mathrm{DT}_{j}}\right\Vert ^{2}\left\Vert \mathrm{DT}_{j}-\mathrm{BS}_{0}\right\Vert ^{-\alpha}$
is the interference from active D2D transmitters using the same OFDMA
tone with $\mathrm{U}_{0}$ .

Similarly, according to (\ref{eq:recv signal at D2D OFDMA}), the
SIR at $\mathrm{DR}_{0}$ using OFDMA in the underlaid mode is given
by
\begin{equation}
\mathrm{SIR}_{\mathrm{DR}_{0}}^{\mathrm{OU}}=\frac{P_{\mathrm{D}}r_{\mathrm{D},0}^{-\alpha}\left\Vert h_{\mathrm{DR}_{0},\mathrm{DT}_{0}}\right\Vert ^{2}}{I_{\mathrm{D},\mathrm{D}}^{\mathrm{O}}+I_{\mathrm{D},\mathrm{C}}^{\mathrm{O}}},\label{eq:SIR at D2D OFDMA}
\end{equation}
where $I_{\mathrm{D},\mathrm{D}}^{\mathrm{O}}=\underset{\tiny{\mathrm{DT}_{j}\in\tilde{\Pi}_{\mathrm{DT}}^{\mathrm{O}}}}{\sum}\frac{P_{\mathrm{D}}\left\Vert h_{\mathrm{DR}_{0},\mathrm{DT}_{j}}\right\Vert ^{2}}{\left\Vert \mathrm{DT}_{j}-\mathrm{DR}_{0}\right\Vert ^{\alpha}}$
is the interference from other D2D transmitters over the same OFDMA
tone with $\mathrm{DR}_{0}$ and $I_{\mathrm{D},\mathrm{C}}^{\mathrm{O}}=\underset{\tiny{\mathrm{U}_{i}\in\Pi_{\mathrm{UT}}^{\mathrm{O}}}}{\sum}\frac{P_{\mathrm{U}}\left\Vert h_{\mathrm{DR}_{0},\mathrm{U}_{i}}\right\Vert ^{2}}{\left\Vert \mathrm{U}_{i}-\mathrm{DR}_{0}\right\Vert ^{\alpha}}$
is the interference from active cellular users over the same OFDMA
tone with $\mathrm{DR}_{0}$.

Based on (\ref{eq:SIR at BS OFDMA}) and (\ref{eq:SIR at D2D OFDMA}),
we derive ASE using the following proposition.

\begin{proposition}

Considering that OFDMA is used in D2D underlaid cellular networks,
the ASE of cellular network is given by $\mathcal{A}_{\mathrm{C}}^{\mathrm{OU}}=q_{\mathrm{U}}^{\mathrm{O}}\lambda_{\mathrm{UT}}\mathrm{CP}_{\mathrm{BS}}^{\mathrm{OU}}\mathrm{log}\left(1+\tau_{\mathrm{BS}}\right)$,
where 
\begin{equation}
\mathrm{CP}_{\mathrm{BS}}^{\mathrm{OU}}=\frac{\lambda_{\mathrm{BS}}}{\lambda_{\mathrm{BS}}+\frac{2q_{\mathrm{U}}^{\mathrm{O}}\lambda_{\mathrm{UT}}\tau_{\mathrm{BS}}}{K\left(\alpha-2\right)}HyF_{1}+\frac{2\pi\lambda_{\mathrm{DT}}\left(\tau_{\mathrm{BS}}\eta_{\mathrm{P}}\right)^{\delta}}{K\alpha\sin\left(\frac{2\pi}{\alpha}\right)}}.\label{eq:CP at BS OFDMA}
\end{equation}
In (\ref{eq:CP at BS OFDMA}), $HyF_{1}={}_{2}F_{1}\left(1,1-\delta,2-\delta,-\tau_{\mathrm{BS}}\right)$,
where $_{2}F_{1}\left(\cdot,\cdot,\cdot,\cdot\right)$ denotes the
hypergeometric function, $\eta_{\mathrm{P}}=\frac{P_{\mathrm{D}}}{P_{\mathrm{U}}}$
and $\delta=\frac{2}{\alpha}$. The D2D network ASE is given by $\mathcal{A}_{\mathrm{D}}^{\mathrm{OU}}=\lambda_{\mathrm{DT}}\mathrm{CP}_{\mathrm{DR}}^{\mathrm{OU}}\mathrm{log}\left(1+\tau_{\mathrm{DR}}\right)$,
where 
\begin{equation}
\mathrm{CP}_{\mathrm{DR}}^{\mathrm{OU}}=\pi\xi\rho_{\mathrm{O}}^{-1}\left(1-e^{-\rho_{\mathrm{O}}\tau_{\mathrm{dis}}^{2}}\right).\label{eq:CP at D2D OFDMA}
\end{equation}
In (\ref{eq:CP at D2D OFDMA}), $\rho_{\mathrm{O}}=\pi\xi+\frac{2\pi^{2}\tau_{\mathrm{DR}}^{\delta}\left(\lambda_{\mathrm{DT}}+q_{\mathrm{U}}^{\mathrm{O}}\lambda_{\mathrm{UT}}\eta_{\mathrm{P}}^{-\delta}\right)}{K\alpha\sin\left(\frac{2\pi}{\alpha}\right)}$.

\label{proposition: OFDMA ASE underlaid}

\end{proposition}

\begin{proof}See Appendix \ref{sub:Proof for Prop cellular CP OFDMA}.\end{proof}

Note that it is easy to obtain $\lambda_{\mathrm{UT}}=\lambda_{\mathrm{U}}+\lambda_{\mathrm{D}}e^{-\xi\pi\tau_{\mathrm{dis}}^{2}}$
and $\lambda_{\mathrm{DT}}=\lambda_{\mathrm{D}}\left(1-e^{-\xi\pi\tau_{\mathrm{dis}}^{2}}\right)$
according to the mode selection rule introduced in Section \ref{sec:System-Model}.
We notice from Proposition \ref{proposition: OFDMA ASE underlaid}
that the ASE achieved by cellular networks increases with the intensity
of active cellular users, i.e., $q_{\mathrm{U}}^{\mathrm{O}}\lambda_{\mathrm{UT}}$.
However, given the number of OFDMA tones and intensity of BSs, $q_{\mathrm{U}}^{\mathrm{O}}\lambda_{\mathrm{UT}}$
cannot be further improved when $\lambda_{\mathrm{UT}}$ is sufficiently
large. Consequently, cellular network ASE is limited. In order to
increase cellular network ASE, it is intuitive to increase the intensity
of active cellular users, which can be accomplished using SCMA.

According to (\ref{eq:recv signal at BS SCMA}), when SCMA is applied,
the SIR at $\mathrm{BS}_{0}$ in the underlaid scenario can be expressed
as
\begin{equation}
\mathrm{SIR}_{\mathrm{BS}_{0}}^{\mathrm{SU}}=\frac{\underset{m=1}{\overset{N_{\mathrm{C}}}{\sum}}P_{\mathrm{U}}^{\dagger}r_{\mathrm{U},0}^{-\alpha}\left\Vert h_{\mathrm{BS}_{0},\mathrm{U}_{0},m}\right\Vert ^{2}}{I_{\mathrm{C},\mathrm{C}}^{\mathrm{S}}+I_{\mathrm{C},\mathrm{D}}^{\mathrm{S}}},\label{eq:SIR at BS SCMA}
\end{equation}
where $I_{\mathrm{C},\mathrm{C}}^{\mathrm{S}}=\underset{\tiny{\mathrm{U}_{i}\in\tilde{\Pi}_{\mathrm{UT}}^{\mathrm{S}}}}{\sum}\overset{N_{\mathrm{C}}}{\underset{m=1}{\sum}}\frac{P_{\mathrm{U}}^{\dagger}\left\Vert h_{\mathrm{BS}_{0},\mathrm{U}_{i},m}\right\Vert ^{2}}{\left\Vert \mathrm{U}_{i}-\mathrm{BS}_{0}\right\Vert ^{\alpha}}$
and $I_{\mathrm{C},\mathrm{D}}^{\mathrm{S}}=\underset{\tiny{\mathrm{DT}_{j}\in\Pi_{\mathrm{DT}}^{\mathrm{S}}}}{\sum}\overset{N_{\mathrm{C}}}{\underset{m=1}{\sum}}\frac{P_{\mathrm{D}}^{\dagger}\left\Vert h_{\mathrm{BS}_{0},\mathrm{DT}_{j},m}\right\Vert ^{2}}{\left\Vert \mathrm{DT}_{j}-\mathrm{BS}_{0}\right\Vert ^{\alpha}}$.

Likewise, based on (\ref{eq:recv signal at D2D SCMA}), the SIR at
$\mathrm{DR}_{0}$ is given by
\begin{equation}
\mathrm{SIR}_{\mathrm{DR}_{0}}^{\mathrm{SU}}=\frac{\overset{N_{\mathrm{C}}}{\underset{m=1}{\sum}}P_{\mathrm{D}}^{\dagger}r_{\mathrm{D},0}^{-\alpha}\left\Vert h_{\mathrm{DR}_{0},\mathrm{DT}_{0},m}\right\Vert ^{2}}{I_{\mathrm{D},\mathrm{D}}^{\mathrm{S}}+I_{\mathrm{D},\mathrm{C}}^{\mathrm{S}}},\label{eq:SIR at D2D SCMA}
\end{equation}
where $I_{\mathrm{D},\mathrm{D}}^{\mathrm{S}}=\underset{\tiny{\mathrm{DT}_{j}\in\tilde{\Pi}_{\mathrm{DT}}^{\mathrm{S}}}}{\sum}\overset{N_{\mathrm{C}}}{\underset{m=1}{\sum}}\frac{P_{\mathrm{D}}^{\dagger}\left\Vert h_{\mathrm{DR}_{0},\mathrm{DT}_{j},m}\right\Vert ^{2}}{\left\Vert \mathrm{DT}_{i}-\mathrm{DR}_{0}\right\Vert ^{\alpha}}$
and $I_{\mathrm{D},\mathrm{C}}^{\mathrm{S}}=\underset{\tiny{\mathrm{U}_{i}\in\Pi_{\mathrm{UT}}^{\mathrm{S}}}}{\sum}\overset{N_{\mathrm{C}}}{\underset{m=1}{\sum}}\frac{P_{\mathrm{U}}^{\dagger}\left\Vert h_{\mathrm{DR}_{0},\mathrm{U}_{i},m}\right\Vert ^{2}}{\left\Vert \mathrm{U}_{i}-\mathrm{DR}_{0}\right\Vert ^{\alpha}}$.

In order to derive ASE, we should first calculate the coverage probabilities
at $\mathrm{BS}_{0}$ and $\mathrm{DR}_{0}$. Note from (\ref{eq:SIR at BS SCMA})
and (\ref{eq:SIR at D2D SCMA}) that the SCMA codeword is spread over
more than one OFDMA tone, which brings about difficulty to obtain
the exact results of coverage probabilities in explicit form. Accordingly,
the methods utilized in \cite{Cognitive_D2D_cellular} cannot be applied.
Targeting at providing design insights through the results, we use
approximations to derive the closed-form expressions of $\mathrm{CP}_{\mathrm{BS}}^{\mathrm{SU}}=\mathbb{P}\left\{ \mathrm{SIR}_{\mathrm{BS}_{0}}^{\mathrm{SU}}>\tau_{\mathrm{BS}}\right\} $
and $\mathrm{CP}_{\mathrm{BS}}^{\mathrm{SU}}=\mathbb{P}\left\{ \mathrm{SIR}_{\mathrm{DR}_{0}}^{\mathrm{SU}}>\tau_{\mathrm{DR}}\right\} $
according to the following proposition.

\begin{proposition}

Considering that SCMA is used in D2D underlaid cellular networks,
the ASE achieved by cellular network is given by $\mathcal{A}_{\mathrm{C}}^{\mathrm{SU}}=q_{\mathrm{U}}^{\mathrm{S}}\lambda_{\mathrm{UT}}\mathrm{CP}_{\mathrm{BS}}^{\mathrm{SU}}\mathrm{log}\left(1+\tau_{\mathrm{BS}}\right)$,
where $\mathrm{CP}_{\mathrm{BS}}^{\mathrm{SU}}$ is approximated as
\begin{equation}
\mathrm{CP}_{\mathrm{BS}}^{\mathrm{SU}}\approx\frac{\lambda_{\mathrm{BS}}}{\lambda_{\mathrm{BS}}+\frac{q_{\mathrm{U}}^{\mathrm{S}}\lambda_{\mathrm{UT}}}{J}HyF_{2}+\frac{2\pi\lambda_{\mathrm{DT}}\left(\tilde{\tau}_{\mathrm{BS}}\eta_{\mathrm{P}}\right)^{\delta}}{J\alpha\sin\left(\frac{2\pi}{\alpha}\right)}\underset{_{n=2}}{\overset{_{N_{\mathrm{C}}}}{\prod}}\left(\frac{2}{\left(n-1\right)\alpha}+1\right)}.\label{eq:CP at BS SCMA}
\end{equation}
In (\ref{eq:CP at BS SCMA}), $HyF_{2}={}_{2}F_{1}\left(N_{\mathrm{C}},-\delta,1-\delta,-\tilde{\tau}_{\mathrm{BS}}\right)-1$
and $\tilde{\tau}_{\mathrm{BS}}=\frac{\tau_{\mathrm{BS}}}{N_{\mathrm{C}}}$.
The ASE achieved by D2D network is given by $\mathcal{A}_{\mathrm{D}}^{\mathrm{SU}}=\lambda_{\mathrm{DT}}\mathrm{CP}_{\mathrm{DR}}^{\mathrm{SU}}\mathrm{log}\left(1+\tau_{\mathrm{DR}}\right)$,
where $\mathrm{CP}_{\mathrm{DR}}^{\mathrm{SU}}$ is approximated as
\begin{equation}
\mathrm{CP}_{\mathrm{DR}}^{\mathrm{SU}}\approx\pi\xi\rho_{\mathrm{S}}^{-1}\left(1-e^{-\rho_{\mathrm{S}}\tau_{\mathrm{dis}}^{2}}\right).\label{eq:CP at D2D SCMA}
\end{equation}
In (\ref{eq:CP at D2D SCMA}), $\rho_{\mathrm{S}}=\pi\xi+\frac{2\pi^{2}\tilde{\tau}_{\mathrm{DR}}^{\delta}\left(\lambda_{\mathrm{DT}}+q_{\mathrm{U}}^{\mathrm{S}}\lambda_{\mathrm{UT}}\eta_{\mathrm{P}}^{-\delta}\right)}{J\alpha\sin\left(\frac{2\pi}{\alpha}\right)}\underset{_{n=2}}{\overset{_{N_{\mathrm{C}}}}{\prod}}\left(\frac{2}{\left(n-1\right)\alpha}+1\right)$
and $\tilde{\tau}_{\mathrm{DR}}=\frac{\tau_{\mathrm{DR}}}{N_{\mathrm{C}}}$.

\label{proposition: SCMA ASE underlaid}

\end{proposition}

\begin{proof}See Appendix \ref{sub:Proof for Prop cellular CP SCMA}.\end{proof}

In Appendix \ref{sub:Proof for Prop cellular CP SCMA}, we have used
$W_{\mathrm{BS}_{0},\mathrm{U}_{0}}\sim\exp\left(N_{\mathrm{C}}^{-1}\right)$
to approximate $G_{\mathrm{BS}_{0},\mathrm{U}_{0}}\sim\mathrm{Gamma}\left(N_{\mathrm{C}},1\right)$
for analytical tractability%
\footnote{Note that the subscripts of $W$ and $G$ have the same meaning as
that for $h$.%
}. $W_{\mathrm{BS}_{0},\mathrm{U}_{0}}$ and $G_{\mathrm{BS}_{0},\mathrm{U}_{0}}$
have the same first moment. Meanwhile, it can be readily shown that
higher moments of $W_{\mathrm{BS}_{0},\mathrm{U}_{0}}$ and $G_{\mathrm{BS}_{0},\mathrm{U}_{0}}$
are close especially when $N_{\mathrm{C}}$ is small. Note that the
typical value of $N_{\mathrm{C}}$ equals 2. Besides, the power loss
caused by the small-scale fading is much smaller than that caused
by pathloss. Therefore, using such approximation will not exert much
influence on the accuracy of the coverage probability analysis provided
in Proposition \ref{proposition: SCMA ASE underlaid}, which will
be shown in Section \ref{sec:Numerical Results}.

Comparing the results given by Proposition \ref{proposition: OFDMA ASE underlaid}
and Proposition \ref{proposition: SCMA ASE underlaid}, we derive
the ASE gain of SCMA over OFDMA as
\begin{equation}
\eta_{\mathrm{ASE}}=\frac{\mathcal{A}_{\mathrm{C}}^{\mathrm{SU}}+\mathcal{A}_{\mathrm{D}}^{\mathrm{SU}}}{\mathcal{A}_{\mathrm{C}}^{\mathrm{OU}}+\mathcal{A}_{\mathrm{D}}^{\mathrm{OU}}}.\label{eq:ASE gain}
\end{equation}
By definition, if the same SIR thresholds are applied by cellular
users and D2D users, i.e., $\tau_{\mathrm{BS}}=\tau_{\mathrm{DR}}$,
the ASE gain is equivalent to
\begin{equation}
\hat{\eta}_{\mathrm{ASE}}=\frac{q_{\mathrm{U}}^{\mathrm{S}}\lambda_{\mathrm{UT}}\mathrm{CP}_{\mathrm{BS}}^{\mathrm{SU}}+\lambda_{\mathrm{DT}}\mathrm{CP}_{\mathrm{DR}}^{\mathrm{SU}}}{q_{\mathrm{U}}^{\mathrm{O}}\lambda_{\mathrm{UT}}\mathrm{CP}_{\mathrm{BS}}^{\mathrm{OU}}+\lambda_{\mathrm{DT}}\mathrm{CP}_{\mathrm{DR}}^{\mathrm{OU}}},\label{eq:ASE gain ex}
\end{equation}
where the numerator denotes the intensity of users that are successfully
admitted using SCMA and the denominator denotes the intensity of users
that are successfully admitted using OFDMA. Therefore, we could compare
$\hat{\eta}_{\mathrm{ASE}}$ with $\eta_{\mathrm{overload}}$ defined
in (\ref{eq:overloading factor}) to examine whether overloading gain
can be achieved when codebook reuse is enabled, which will be shown
in Section \ref{sec:Numerical Results}.

\subsection{Optimizing Activated Probability of D2D Users}

In this part, we enable D2D users to use an activated probability
to contend for available resources. In particular, at the beginning
of each time slot, each D2D transmitter randomly and independently
tosses a coin to determine whether to keep active during this time
slot. Then, the activated D2D transmitters use SCMA for data transmission.
Given that the activated probability of each D2D transmitters is $q_{\mathrm{D}}$,
the set of active D2D transmitters is a thinned PPP with intensity
$q_{\mathrm{D}}\lambda_{\mathrm{DT}}$. Hence, $q_{\mathrm{D}}$ is
a tunable parameter to control cross-tier interference from D2D network
to cellular network. Specifically, a large $q_{\mathrm{D}}$ would
enhance D2D network performance by activating more D2D transmissions,
but result in overwhelming cross-tier interference to cellular transmissions,
and vice versa. Meanwhile, D2D transmissions are fully under control
of BSs. Specifically, BSs decide how many D2D users are active through
tuning the activated probability. After the activated probability
is set, this scheme can be independently performed by D2D users without
channel estimation and channel information interaction. Therefore,
this scheme serves as a simple but effective semi-centralized interference
management method. Although suboptimal compared to conventional centralized
scheme, it can provide a performance lower bound to these sophisticated
schemes.

In the following, we intend to balance the tradeoff by optimizing
$q_{\mathrm{D}}$ based on a proportional fairness utility function.
To this end, we first quantify the effect of $q_{\mathrm{D}}$ on
the hybrid system performance according to the following corollary.

\begin{corollary}

\noindent Considering that D2D transmitters are activated with probability
$q_{\mathrm{D}}$ in the underlaid scenario, the ASE of cellular network
using SCMA is given by $\hat{\mathcal{A}}_{\mathrm{C}}^{\mathrm{SU}}=q_{\mathrm{U}}^{\mathrm{S}}\lambda_{\mathrm{UT}}\hat{\mathrm{CP}}_{\mathrm{BS}}^{\mathrm{SU}}\mathrm{log}\left(1+\tau_{\mathrm{BS}}\right)$,
where $\hat{\mathrm{CP}}_{\mathrm{BS}}^{\mathrm{SU}}$ is approximated
as
\begin{equation}
\hat{\mathrm{CP}}_{\mathrm{BS}}^{\mathrm{SU}}\approx\frac{\lambda_{\mathrm{BS}}}{\lambda_{\mathrm{BS}}+\frac{q_{\mathrm{U}}^{\mathrm{S}}\lambda_{\mathrm{UT}}}{J}HyF_{2}+\frac{2\pi q_{\mathrm{D}}\lambda_{\mathrm{DT}}\left(\tilde{\tau}_{\mathrm{BS}}\eta_{\mathrm{P}}\right)^{\delta}}{J\alpha\sin\left(\frac{2\pi}{\alpha}\right)}\underset{_{n=2}}{\overset{_{N_{\mathrm{C}}}}{\prod}}\left(\frac{2}{\left(n-1\right)\alpha}+1\right)}.\label{eq:CP at BS SCMA MAP}
\end{equation}
The ASE of D2D network is given by $\hat{\mathcal{A}}_{\mathrm{D}}^{\mathrm{SU}}=q_{\mathrm{D}}\lambda_{\mathrm{DT}}\hat{\mathrm{CP}}_{\mathrm{DR}}^{\mathrm{SU}}\mathrm{log}\left(1+\tau_{\mathrm{DR}}\right)$,
where $\hat{\mathrm{CP}}_{\mathrm{DR}}^{\mathrm{SU}}$ is approximated
as
\begin{equation}
\hat{\mathrm{CP}}_{\mathrm{DR}}^{\mathrm{SU}}\approx\frac{\pi\xi}{\rho_{\mathrm{S}}^{\dagger}}\left(1-e^{-\rho_{\mathrm{S}}^{\dagger}\tau_{\mathrm{dis}}^{2}}\right).\label{eq:CP at D2D SCMA MAP}
\end{equation}
In (\ref{eq:CP at D2D SCMA MAP}), $\rho_{\mathrm{S}}^{\dagger}=\pi\xi+\frac{2\pi^{2}\tilde{\tau}_{\mathrm{DR}}^{\delta}\left(q_{\mathrm{D}}\lambda_{\mathrm{DT}}+q_{\mathrm{U}}^{\mathrm{S}}\lambda_{\mathrm{UT}}\eta_{\mathrm{P}}^{-\delta}\right)}{J\alpha\sin\left(\frac{2\pi}{\alpha}\right)}\underset{_{n=2}}{\overset{_{N_{\mathrm{C}}}}{\prod}}\left(\frac{2}{\left(n-1\right)\alpha}+1\right)$.

\label{corollary: ASE using D2D MAP}

\end{corollary}

\begin{proof}The proof is similar to that in Appendix \ref{sub:Proof for Prop cellular CP SCMA}
and thus is omitted due to space limitation.\end{proof}

We then search for the optimal activated probability $q_{\mathrm{D}}^{*}$
according to the following objective function
\begin{align}
q_{\mathrm{D}}^{*} & =\arg\underset{q_{\mathrm{D}}}{\max}\: u_{\mathrm{U}}\left(\hat{\mathcal{A}}_{\mathrm{C}}^{\mathrm{SU}},\hat{\mathcal{A}}_{\mathrm{D}}^{\mathrm{SU}}\right),\label{eq:original problem underlaid}
\end{align}
where $q_{\mathrm{D}}\in\left[0,1\right]$ and $u_{\mathrm{U}}\left(\hat{\mathcal{A}}_{\mathrm{C}}^{\mathrm{SU}},\hat{\mathcal{A}}_{\mathrm{D}}^{\mathrm{SU}}\right)$
is a utility function that is in different forms according to different
design targets. In our work, we define the utility function based
on the most commonly used proportional fairness, i.e.,
\begin{equation}
u_{\mathrm{U}}\left(\hat{\mathcal{A}}_{\mathrm{C}}^{\mathrm{SU}},\hat{\mathcal{A}}_{\mathrm{D}}^{\mathrm{SU}}\right)=\log\hat{\mathcal{A}}_{\mathrm{C}}^{\mathrm{SU}}+\log\hat{\mathcal{A}}_{\mathrm{D}}^{\mathrm{SU}},\label{eq:utility function underlaid}
\end{equation}
where $\hat{\mathcal{A}}_{\mathrm{C}}^{\mathrm{SU}}$ and $\hat{\mathcal{A}}_{\mathrm{D}}^{\mathrm{SU}}$
are given by Corollary \ref{corollary: ASE using D2D MAP}. According
to Corollary \ref{corollary: ASE using D2D MAP}, it is observed that
$q_{\mathrm{D}}$ has a complicated influence on $\hat{\mathcal{A}}_{\mathrm{C}}^{\mathrm{SU}}$
and $\hat{\mathcal{A}}_{\mathrm{D}}^{\mathrm{SU}}$. Although $q_{\mathrm{D}}^{*}$
can be numerically obtained, a closed-form solution is unobtainable.
In the following, we consider two special cases such that the approximate
$q_{\mathrm{D}}^{*}$ can be derived in closed-form.

We first consider $\rho_{\mathrm{S}}^{\dagger}\tau_{\mathrm{dis}}^{2}$
in Corollary \ref{corollary: ASE using D2D MAP} satisfies $\rho_{\mathrm{S}}^{\dagger}\tau_{\mathrm{dis}}^{2}\rightarrow0$.
This corresponds to the case, where $\xi$ is small%
\footnote{According to (\ref{eq:D2D link length distribution}), $\mathbb{E}\left[r_{\mathrm{D},k}\right]=\frac{1}{2\sqrt{\xi}}$
such that $\mathbb{E}\left[r_{\mathrm{D},k}\right]$ scales inversely
with $\sqrt{\xi}$. For instance, when $\xi=5\times10^{-5}$, $\mathbb{E}\left[r_{\mathrm{D},k}\right]\approx70.7\:\mathrm{m}$.
Therefore, the value of $\xi$ is usually small due to the proximity
nature of D2D transmissions.%
}, the intensities of cellular users and D2D transmitters are small
and mode selection threshold is small. Under this condition, we give
the approximate $q_{\mathrm{D}}^{*}$ in the following theorem.

\begin{theorem}

\noindent In SCMA enhanced D2D underlaid cellular network, when $\rho_{\mathrm{S}}^{\dagger}\tau_{\mathrm{dis}}^{2}\rightarrow0$,
the optimal activated probability that maximizes the proportional
fairness utility function is $q_{\mathrm{D}}^{*}=1$.

\label{theorem: optimize MAP approximation}

\end{theorem}

\begin{proof}According to Corollary \ref{corollary: ASE using D2D MAP},
when $\rho_{\mathrm{S}}^{\dagger}\tau_{\mathrm{dis}}^{2}\rightarrow0$,
the coverage probability at the typical D2D receiver using SCMA is
approximated as $\hat{\mathrm{CP}}_{\mathrm{DR}}^{\mathrm{SU}}=\pi\xi\tau_{\mathrm{dis}}^{2}$,
which is due to the first order Taylor expansion of $e^{-\rho_{\mathrm{S}}^{\dagger}\tau_{\mathrm{dis}}^{2}}$.
Therefore, the objective function defined in (\ref{eq:original problem underlaid})
turns into
\[
q_{\mathrm{D}}^{*}=\arg\underset{q_{\mathrm{D}}}{\max}\:\log\left[\frac{q_{\mathrm{U}}^{\mathrm{S}}\lambda_{\mathrm{U}}\lambda_{\mathrm{BS}}\lambda_{\mathrm{DT}}\pi\xi\tau_{\mathrm{dis}}^{2}q_{\mathrm{D}}}{Q_{1}+Q_{2}q_{\mathrm{D}}}\mathrm{log}\left(1+\tau_{\mathrm{BS}}\right)\mathrm{log}\left(1+\tau_{\mathrm{DR}}\right)\right],
\]
where we denote $Q_{1}=\lambda_{\mathrm{BS}}+\frac{q_{\mathrm{U}}^{\mathrm{S}}\lambda_{\mathrm{UT}}}{J}HyF_{2}$
and $Q_{2}=\frac{2\pi\lambda_{\mathrm{DT}}\left(\tilde{\tau}_{\mathrm{BS}}\eta_{\mathrm{P}}\right)^{\delta}}{J\alpha\sin\left(\frac{2\pi}{\alpha}\right)}\underset{_{n=2}}{\overset{_{N_{\mathrm{C}}}}{\prod}}\left(\frac{2}{\left(n-1\right)\alpha}+1\right)$
for notation simplicity. As $\log\left(\cdot\right)$ is a monotonically
increasing function, $q_{\mathrm{D}}^{*}$ can be obtained by searching
for the optimal $q_{\mathrm{D}}$ to maximize $\bar{u}_{\mathrm{U}}\left(q_{\mathrm{D}}\right)=\frac{q_{\mathrm{D}}}{Q_{1}+Q_{2}q_{\mathrm{D}}}$.
It is easy to prove that $\nabla\bar{u}_{\mathrm{U}}\left(q_{\mathrm{D}}\right)>0$
and $\nabla^{2}\bar{u}_{\mathrm{U}}\left(q_{\mathrm{D}}\right)<0$
when $q_{\mathrm{D}}\in\left[0,1\right]$, where $\nabla\bar{u}_{\mathrm{U}}\left(q_{\mathrm{D}}\right)$
and $\nabla^{2}\bar{u}_{\mathrm{U}}\left(q_{\mathrm{D}}\right)$ denote
the first and second derivatives of $\bar{u}_{\mathrm{U}}$, respectively.
Therefore, $\bar{u}_{\mathrm{U}}\left(q_{\mathrm{D}}\right)$ is a
concave function and an increasing function of $q_{\mathrm{D}}$.
As $q_{\mathrm{D}}\in\left[0,1\right]$, the optimal $q_{\mathrm{D}}$
to maximize $\bar{u}_{\mathrm{U}}\left(q_{\mathrm{D}}\right)$ equals
1.\end{proof}

It can be seen from Theorem \ref{theorem: optimize MAP approximation}
that when $\rho_{\mathrm{S}}^{\dagger}\tau_{\mathrm{dis}}^{2}\rightarrow0$,
all the D2D transmitters should contend for available resources such
that the proportional fairness utility function can be maximized.
This coincides with the intuition. In particular, when cellular users
and D2D users are sparsely distributed, spectrum resources cannot
be fully exploited. Therefore, enabling full access of D2D users can
boost the spectrum utilization of D2D network without significantly
degenerating the performance of cellular network.

Next, we consider $\tau_{\mathrm{dis}}\rightarrow\infty$. In this
case, all the D2D users select D2D mode for data transmission. The
optimal $q_{\mathrm{D}}$ to maximize (\ref{eq:original problem underlaid})
is given by the following theorem.

\begin{theorem}

\noindent In SCMA enhanced D2D underlaid cellular network, when $\tau_{\mathrm{dis}}\rightarrow\infty$,
the optimal activated probability to maximize (\ref{eq:original problem underlaid})
is given by
\begin{equation}
q_{\mathrm{D}}^{*}=\left\{ \begin{array}{cc}
\sqrt{\frac{Q_{1}Q_{4}}{Q_{2}Q_{3}}}, & \mathrm{if}\: Q_{1}Q_{4}<Q_{2}Q_{3},\\
1, & \mathrm{if}\: Q_{1}Q_{4}\geq Q_{2}Q_{3},
\end{array}\right.\label{eq:optimal MAP}
\end{equation}
where $Q_{3}=\frac{2\pi\tilde{\tau}_{\mathrm{DR}}^{\delta}\lambda_{\mathrm{DT}}}{J\alpha\sin\left(\frac{2\pi}{\alpha}\right)}\underset{_{n=2}}{\overset{_{N_{\mathrm{C}}}}{\prod}}\left(\frac{2}{\left(n-1\right)\alpha}+1\right)$
and $Q_{4}=\xi+\frac{2\pi\tilde{\tau}_{\mathrm{DR}}^{\delta}q_{\mathrm{U}}^{\mathrm{S}}\lambda_{\mathrm{UT}}\eta_{\mathrm{P}}^{-\delta}}{J\alpha\sin\left(\frac{2\pi}{\alpha}\right)}\underset{_{n=2}}{\overset{_{N_{\mathrm{C}}}}{\prod}}\left(\frac{2}{\left(n-1\right)\alpha}+1\right)$.

\label{theorem: optimize MAP no mode selection}

\end{theorem}

\begin{proof}When $\tau_{\mathrm{dis}}\rightarrow\infty$, all the
D2D users select D2D mode, i.e., $\mathbb{P}\left(r_{\mathrm{D},k}\leq\tau_{\mathrm{dis}}\right)=1$.
According to Corollary \ref{corollary: ASE using D2D MAP}, D2D network
ASE degenerates into $\hat{\mathcal{A}}_{\mathrm{D}}^{\mathrm{SU}}=\frac{\lambda_{\mathrm{DT}}\xi}{Q_{3}+Q_{4}q_{\mathrm{D}}^{-1}}\mathrm{log}\left(1+\tau_{\mathrm{DR}}\right)$
and cellular network ASE is given by $\hat{\mathcal{A}}_{\mathrm{C}}^{\mathrm{SU}}=\frac{q_{\mathrm{U}}^{\mathrm{S}}\lambda_{\mathrm{UT}}\lambda_{\mathrm{BS}}}{Q_{1}+Q_{2}q_{\mathrm{D}}}\mathrm{log}\left(1+\tau_{\mathrm{BS}}\right)$.
Since $\log\left(\cdot\right)$ is a monotonically increasing function,
the original optimization problem (\ref{eq:original problem underlaid})
can be converted into $q_{\mathrm{D}}^{*}=\arg\underset{q_{\mathrm{D}}}{\min}\tilde{u}_{\mathrm{U}}\left(q_{\mathrm{D}}\right)$,
where $\tilde{u}_{\mathrm{U}}\left(q_{\mathrm{D}}\right)=\left(Q_{1}+Q_{2}q_{\mathrm{D}}\right)\left(Q_{3}+Q_{4}q_{\mathrm{D}}^{-1}\right)$.
We obtain the second derivative of $\tilde{u}_{\mathrm{U}}\left(q_{\mathrm{D}}\right)$
as $\nabla^{2}\tilde{u}_{\mathrm{U}}\left(q_{\mathrm{D}}\right)=\frac{2Q_{1}Q_{4}}{q_{\mathrm{D}}^{3}}$.
Since $Q_{1}>0$ and $Q_{4}>0$, it is obvious that $\nabla^{2}\tilde{u}_{\mathrm{U}}\left(q_{\mathrm{D}}\right)>0$
when $q_{\mathrm{D}}\in\left[0,1\right]$. According to \cite{book_convex_optimization},
$\tilde{u}_{\mathrm{U}}\left(q_{\mathrm{D}}\right)$ is a convex function
of $q_{\mathrm{D}}$. Consequently, the results in (\ref{eq:optimal MAP})
can be obtained by solving $\nabla\tilde{u}_{\mathrm{U}}\left(q_{\mathrm{D}}\right)=0$.\end{proof}

It is worth noting from Theorem \ref{theorem: optimize MAP no mode selection}
that $q_{\mathrm{D}}^{*}$ is an increasing function of the number
of available codebooks $J$. We intuitively explain the reason as
follows. As earlier discussed, $q_{\mathrm{D}}$ is a parameter that
can balance the cross-tier interference from D2D transmissions to
cellular transmissions. According to Corollary \ref{corollary: ASE using D2D MAP},
increasing $q_{\mathrm{D}}$ will improve the D2D network ASE. However,
more cross-tier interference will be introduced, resulting in the
decrease of cellular network ASE. Therefore, the proportional fairness
utility function defined in (\ref{eq:utility function underlaid})
cannot be maximized by excessively increasing $q_{\mathrm{D}}$ when
the interference from D2D network is overwhelming. Given small $J$,
the number of D2D users that reuse one codebook with cellular users
is large. Consequently, $q_{\mathrm{D}}$ has be kept small in order
to mitigate the cross-tier interference, thereby maximizing (\ref{eq:utility function underlaid}).
In contrast, when $J$ is large, the interference becomes ``sparse''
over one codebook such that $q_{\mathrm{D}}$ can be moderately increased
to maximize (\ref{eq:utility function underlaid}).

\section{D2D Overlaid Cellular Network\label{sec:Overlaid D2D}}

In this section, we consider the D2D overlaid cellular network, where
$J_{\mathrm{C}}$ and $J_{\mathrm{D}}$ SCMA codebooks are allocated
to cellular users and D2D users, respectively. Different from underlaid
mode, no cross-tier interference exists between cellular users and
D2D users in the overlaid mode. Therefore, codebook allocation becomes
more dominant in affecting the hybrid system performance. Here, we
intend to find the optimal codebook allocation rule based on a given
utility function
\begin{equation}
J_{\mathrm{C}}^{*}=\arg\underset{J_{\mathrm{C}}}{\max}\: u_{\mathrm{O}}\left(\mathcal{A}_{\mathrm{C}}^{\mathrm{SO}},\mathcal{A}_{\mathrm{D}}^{\mathrm{SO}}\right),\label{eq:original problem overlaid}
\end{equation}
where $\mathcal{A}_{\mathrm{C}}^{\mathrm{SO}}$ and $\mathcal{A}_{\mathrm{D}}^{\mathrm{SO}}$
are the ASEs of cellular network and D2D network, respectively, and
\begin{equation}
u_{\mathrm{O}}\left(\mathcal{A}_{\mathrm{C}}^{\mathrm{SO}},\mathcal{A}_{\mathrm{D}}^{\mathrm{SO}}\right)=\log\mathcal{A}_{\mathrm{C}}^{\mathrm{SO}}+\log\mathcal{A}_{\mathrm{D}}^{\mathrm{SO}}.\label{eq:utility function overlaid}
\end{equation}
Note that optimizing $J_{\mathrm{C}}$ is equivalent to optimizing
$J_{\mathrm{D}}$ in (\ref{eq:original problem overlaid}), since
$J=J_{\mathrm{C}}+J_{\mathrm{D}}$. In order to achieve the optimization
goal, we first determine $\mathcal{A}_{\mathrm{C}}^{\mathrm{SO}}$
and $\mathcal{A}_{\mathrm{D}}^{\mathrm{SO}}$ according to the following
corollary.

\begin{corollary}

Considering that SCMA is used in D2D overlaid cellular network, the
ASE achieved by cellular network is given by $\mathcal{A}_{\mathrm{C}}^{\mathrm{SO}}=\hat{q}_{\mathrm{U}}^{\mathrm{S}}\lambda_{\mathrm{UT}}\mathrm{CP}_{\mathrm{BS}}^{\mathrm{SO}}\mathrm{log}\left(1+\tau_{\mathrm{BS}}\right)$,
where $\mathrm{CP}_{\mathrm{BS}}^{\mathrm{SO}}$ is approximated as
\begin{equation}
\mathrm{CP}_{\mathrm{BS}}^{\mathrm{SO}}\approx\frac{\lambda_{\mathrm{BS}}}{\lambda_{\mathrm{BS}}+\frac{\hat{q}_{\mathrm{U}}^{\mathrm{S}}\lambda_{\mathrm{UT}}}{J_{\mathrm{C}}}HyF_{2}}.\label{eq:CP at BS SCMA overlaid}
\end{equation}
In (\ref{eq:CP at BS SCMA overlaid}), $\hat{q}_{\mathrm{U}}^{\mathrm{S}}$
is calculated by substituting $N_{\mathrm{R}}$ with $J_{\mathrm{C}}$
in (\ref{eq:cellular access probability}). The ASE achieved by D2D
network is given by $\mathcal{A}_{\mathrm{D}}^{\mathrm{SO}}=\lambda_{\mathrm{DT}}\mathrm{CP}_{\mathrm{DR}}^{\mathrm{SO}}\mathrm{log}\left(1+\tau_{\mathrm{DR}}\right)$,
where $\mathrm{CP}_{\mathrm{DR}}^{\mathrm{SO}}$ is approximated as
\begin{equation}
\mathrm{CP}_{\mathrm{DR}}^{\mathrm{SO}}\approx\frac{\pi\xi}{\rho_{\mathrm{O}}}\left(1-e^{-\rho_{\mathrm{O}}\tau_{\mathrm{dis}}^{2}}\right).\label{eq:CP at D2D SCMA overlaid}
\end{equation}
In (\ref{eq:CP at D2D SCMA overlaid}), $\rho_{\mathrm{O}}=\pi\xi+\frac{2\pi^{2}\tilde{\tau}_{\mathrm{DR}}^{\delta}q_{\mathrm{D}}\lambda_{\mathrm{DT}}}{J_{\mathrm{D}}\alpha\sin\left(\frac{2\pi}{\alpha}\right)}\underset{_{n=2}}{\overset{_{N_{\mathrm{C}}}}{\prod}}\left(\frac{2}{\left(n-1\right)\alpha}+1\right)$.

\label{corollary: SCMA ASE overlaid}

\end{corollary}

\begin{proof}We describe the sketch of the proof. When codebooks
are partially allocated to cellular users and D2D users, respectively,
no cross-tier interference exists between them. Therefore, the SIR
at $\mathrm{BS}_{0}$ is given by
\[
\mathrm{SIR}_{\mathrm{BS}_{0}}^{\mathrm{SU}}=\frac{\underset{m=1}{\overset{N_{\mathrm{C}}}{\sum}}P_{\mathrm{U}}^{\dagger}r_{\mathrm{U},0}^{-\alpha}\left\Vert h_{\mathrm{BS}_{0},\mathrm{U}_{0},m}\right\Vert ^{2}}{\tilde{I}_{\mathrm{C},\mathrm{C}}^{\mathrm{S}}},
\]
where $\tilde{I}_{\mathrm{C},\mathrm{C}}^{\mathrm{S}}$ is the inter-cell
interference from active cellular users outside the coverage of $\mathrm{BS}_{0}$.
Note that $\tilde{I}_{\mathrm{C},\mathrm{C}}^{\mathrm{S}}$ is different
from $I_{\mathrm{C},\mathrm{C}}^{\mathrm{S}}$ in (\ref{eq:SIR at BS SCMA}),
as only $J_{\mathrm{C}}$ codebooks are allocated to cellular users.
Using similar approach as (\ref{eq:Laplace of ICC SCMA}) in Appendix
\ref{sub:Proof for Prop cellular CP SCMA}, the coverage probability
in (\ref{eq:CP at BS SCMA overlaid}) can be obtained and hence cellular
network ASE is derived. Similarly, D2D network ASE can be derived
as well. The detail of the derivation steps is omitted due to space
limitation.\end{proof}

According to (\ref{eq:cellular access probability}), $\hat{q}_{\mathrm{U}}^{\mathrm{S}}$
in (\ref{eq:CP at BS SCMA overlaid}) is a function of $J_{\mathrm{C}}$.
Therefore, the influence of codebook allocation in the overlaid scenario
also has complicated impact on the ASE. In consequence, it is hard
to obtain the closed-form expression for $J_{\mathrm{C}}^{*}$. Next,
we consider the case, where cellular users are densely deployed such
that $\lambda_{\mathrm{UT}}\gg J_{\mathrm{C}}\lambda_{\mathrm{BS}}$.
Note that it is well accepted that the dense scenario is a common
scenario in future wireless networks. Under this condition, a proportion
of cellular users would be inactivated due to the unavailability of
SCMA codebooks even if all the codebooks are allocated to cellular
network. In this case, the density of active cellular users is $J_{\mathrm{C}}\lambda_{\mathrm{BS}}$.
Replacing $\hat{q}_{\mathrm{U}}^{\mathrm{S}}\lambda_{\mathrm{U}}$
with $J_{\mathrm{C}}\lambda_{\mathrm{BS}}$ in Corollary \ref{corollary: SCMA ASE overlaid},
cellular network ASE degenerates into
\begin{equation}
\tilde{\mathcal{A}}_{\mathrm{C}}^{\mathrm{SO}}=\frac{J_{\mathrm{C}}\lambda_{\mathrm{BS}}\mathrm{log}\left(1+\tau_{\mathrm{BS}}\right)}{HyF_{2}+1}.\label{eq:ASE at BS SCMA overlaid}
\end{equation}
Besides, if D2D users select cellular mode for data transmission,
they will be blocked with a higher probability due to the dense deployment
of cellular users. Therefore, we enable all the D2D users to select
D2D mode by setting $\tau_{\mathrm{dis}}\rightarrow\infty$. In this
case, $J_{\mathrm{C}}^{*}$ can be derived according to the following
theorem.

\begin{theorem}

\noindent We consider an SCMA enhanced D2D overlaid cellular network,
where cellular users are densely deployed. When $\tau_{\mathrm{dis}}\rightarrow\infty$,
the optimal $J_{\mathrm{C}}$ that maximizes the proportional fairness
utility function defined in (\ref{eq:utility function overlaid})
is given by
\begin{equation}
J_{\mathrm{C}}^{*}=\mathrm{round}\left(J+Q_{6}-\sqrt{Q_{6}^{2}+JQ_{6}}\right),\label{eq:optimal codebook allocation factor}
\end{equation}
where $Q_{6}=\frac{2\pi\tilde{\tau}_{\mathrm{DR}}^{\delta}\lambda_{\mathrm{DT}}}{\xi\alpha\sin\left(\frac{2\pi}{\alpha}\right)}\underset{_{n=2}}{\overset{_{N_{\mathrm{C}}}}{\prod}}\left(\frac{2}{\left(n-1\right)\alpha}+1\right)$
and $\mathrm{round}\left(\cdot\right)$ denotes the round operation.

\label{theorem: optimize codebook allocation}

\end{theorem}

\begin{proof}According to Corollary \ref{corollary: SCMA ASE overlaid},
when $\tau_{\mathrm{dis}}\rightarrow\infty$, $\mathrm{CP}_{\mathrm{DR}}^{\mathrm{SO}}$
in (\ref{eq:CP at D2D SCMA overlaid}) degenerates into
\begin{equation}
\mathrm{CP}_{\mathrm{DR}}^{\mathrm{SO}}=\frac{\xi}{\xi+\frac{2\pi\tilde{\tau}_{\mathrm{DR}}^{\delta}q_{\mathrm{D}}\lambda_{\mathrm{DT}}}{J_{\mathrm{D}}\alpha\sin\left(\frac{2\pi}{\alpha}\right)}\underset{_{n=2}}{\overset{_{N_{\mathrm{C}}}}{\prod}}\left(\frac{2}{\left(n-1\right)\alpha}+1\right)}.\label{eq:CP at D2D SCMA overlaid special case}
\end{equation}

Based on (\ref{eq:ASE at BS SCMA overlaid}) and (\ref{eq:CP at D2D SCMA overlaid special case}),
the utility function in (\ref{eq:utility function overlaid}) turns
into
\begin{equation}
u_{\mathrm{O}}\left(\tilde{\mathcal{A}}_{\mathrm{C}}^{\mathrm{SO}},\tilde{\mathcal{A}}_{\mathrm{D}}^{\mathrm{SO}}\right)=\log\tilde{\mathcal{A}}_{\mathrm{C}}^{\mathrm{SO}}+\log\tilde{\mathcal{A}}_{\mathrm{D}}^{\mathrm{SO}},\label{eq:utility function overlaid 1}
\end{equation}
where $\tilde{\mathcal{A}}_{\mathrm{C}}^{\mathrm{SO}}=J_{\mathrm{C}}Q_{5}$
and $\tilde{\mathcal{A}}_{\mathrm{D}}^{\mathrm{SO}}=\frac{\lambda_{\mathrm{DT}}\mathrm{log}\left(1+\tau_{\mathrm{DR}}\right)}{1+\frac{Q_{6}}{J-J_{\mathrm{C}}}}$.
For notation simplicity, we denote $Q_{5}=\frac{\lambda_{\mathrm{BS}}\mathrm{log}\left(1+\tau_{\mathrm{BS}}\right)}{HyF_{2}+1}$
and $Q_{6}=\frac{2\pi\tilde{\tau}_{\mathrm{DR}}^{\delta}\lambda_{\mathrm{DT}}}{\xi\alpha\sin\left(\frac{2\pi}{\alpha}\right)}\underset{_{n=2}}{\overset{_{N_{\mathrm{C}}}}{\prod}}\left(\frac{2}{\left(n-1\right)\alpha}+1\right)$.
Accordingly, we rewrite (\ref{eq:utility function overlaid 1}) as
\begin{equation}
u_{\mathrm{O}}\left(\tilde{\mathcal{A}}_{\mathrm{C}}^{\mathrm{SO}},\tilde{\mathcal{A}}_{\mathrm{D}}^{\mathrm{SO}}\right)=\log\left(\frac{Q_{5}\lambda_{\mathrm{DT}}\mathrm{log}\left(1+\tau_{\mathrm{DR}}\right)J_{\mathrm{C}}}{1+\frac{Q_{6}}{J-J_{\mathrm{C}}}}\right).\label{eq:utility function overlaid 2}
\end{equation}

In (\ref{eq:utility function overlaid 2}), $Q_{5}\lambda_{\mathrm{DT}}\mathrm{log}\left(1+\tau_{\mathrm{DR}}\right)>0$
and $\log\left(\cdot\right)$ is a monotonically increasing function.
Therefore, maximizing (\ref{eq:utility function overlaid 2}) is equivalently
to maximizing $\tilde{u}_{\mathrm{O}}\left(J_{\mathrm{C}}\right)=\frac{J_{\mathrm{C}}}{1+\frac{Q_{6}}{J-J_{\mathrm{C}}}}$.
We then calculate the second derivative of $\tilde{u}_{\mathrm{O}}\left(J_{\mathrm{C}}\right)$
as $\nabla^{2}\tilde{u}_{\mathrm{O}}\left(J_{\mathrm{C}}\right)=-\frac{2Q_{6}\left(J+Q_{6}\right)}{\left(1+\frac{Q_{6}}{J-J_{\mathrm{C}}}\right)^{3}\left(J-J_{\mathrm{C}}\right)^{3}}$.
It is easy to check that $\nabla^{2}\tilde{u}_{\mathrm{O}}\left(J_{\mathrm{C}}\right)<0$
when $J_{\mathrm{C}}\in\left[0,J\right]$. Hence, $\tilde{u}_{\mathrm{O}}\left(J_{\mathrm{C}}\right)$
is concave function of $J_{\mathrm{C}}$. The optimal $J_{\mathrm{C}}$
to maximize $\tilde{u}_{\mathrm{O}}\left(J_{\mathrm{C}}\right)$ can
be derived by solving $\nabla\tilde{u}_{\mathrm{O}}\left(J_{\mathrm{C}}\right)=0$.\end{proof}

Note from (\ref{eq:optimal codebook allocation factor}) in Theorem
\ref{theorem: optimize codebook allocation} that the optimal $J_{\mathrm{C}}$
that maximizes the proportional fairness utility function is a function
of parameters of D2D network rather than those of cellular network.
This is due to the assumption that cellular users are densely deployed.
According to (\ref{eq:optimal codebook allocation factor}), cellular
network ASE $\tilde{\mathcal{A}}_{\mathrm{C}}^{\mathrm{SO}}$ increases
linearly with $J_{\mathrm{C}}$. In other words, once one more codebook
is provided, the codebook can be fully utilized by the ultra-densely
deployed cellular users, thereby linearly improving $\tilde{\mathcal{A}}_{\mathrm{C}}^{\mathrm{SO}}$.
Therefore, cellular network parameters will not affect the optimal
$J_{\mathrm{C}}$ when we target at maximizing the utility function
defined based on proportional fairness in (\ref{eq:utility function overlaid}).

Note that the codebook allocation rule, which is studied in this section,
is a non-real-time method, since we have used the tools of stochastic
geometry to evaluate the average system performance. Therefore, designing
a real-time method is out of the scope of this paper.

\section{Numerical Results \label{sec:Numerical Results}}

Numerical results and Monte Carlo simulation results have been provided
in this section to demonstrate the performance of the SCMA enhanced
D2D and cellular hybrid network. Default system parameters are set
as $K=20$, $N_{\mathrm{C}}=2$, $\alpha=4$, $\tau_{\mathrm{BS}}=\tau_{\mathrm{DR}}=10\:\mathrm{dB}$,
$P_{\mathrm{U}}=20\:\mathrm{dBmW}$ and $P_{\mathrm{D}}=20\:\mathrm{dBmW}$.
Note that although we consider that $K=20$ OFDMA tones are available,
it is practically infeasible to design $\mathrm{C}_{20}^{2}$ codebooks
due to the limitation of number of available constellations. Nonetheless,
it is practical to design 6 codebooks out of 4 OFDMA tones according
to the mapping shown in Fig. \ref{fig:Mapping SCMA OFDMA}. Therefore,
$J=30$ codebooks are available in simulations. Note that numerical
results and simulation results are drawn by lines and markers, respectively,
in the following figures.

\begin{figure}[t]
\begin{centering}
\subfloat[\label{fig:OFDMA vs SCMA varying cellular intensity}Varying intensity
of cellular users.]{\begin{centering}
\includegraphics[width=3in]{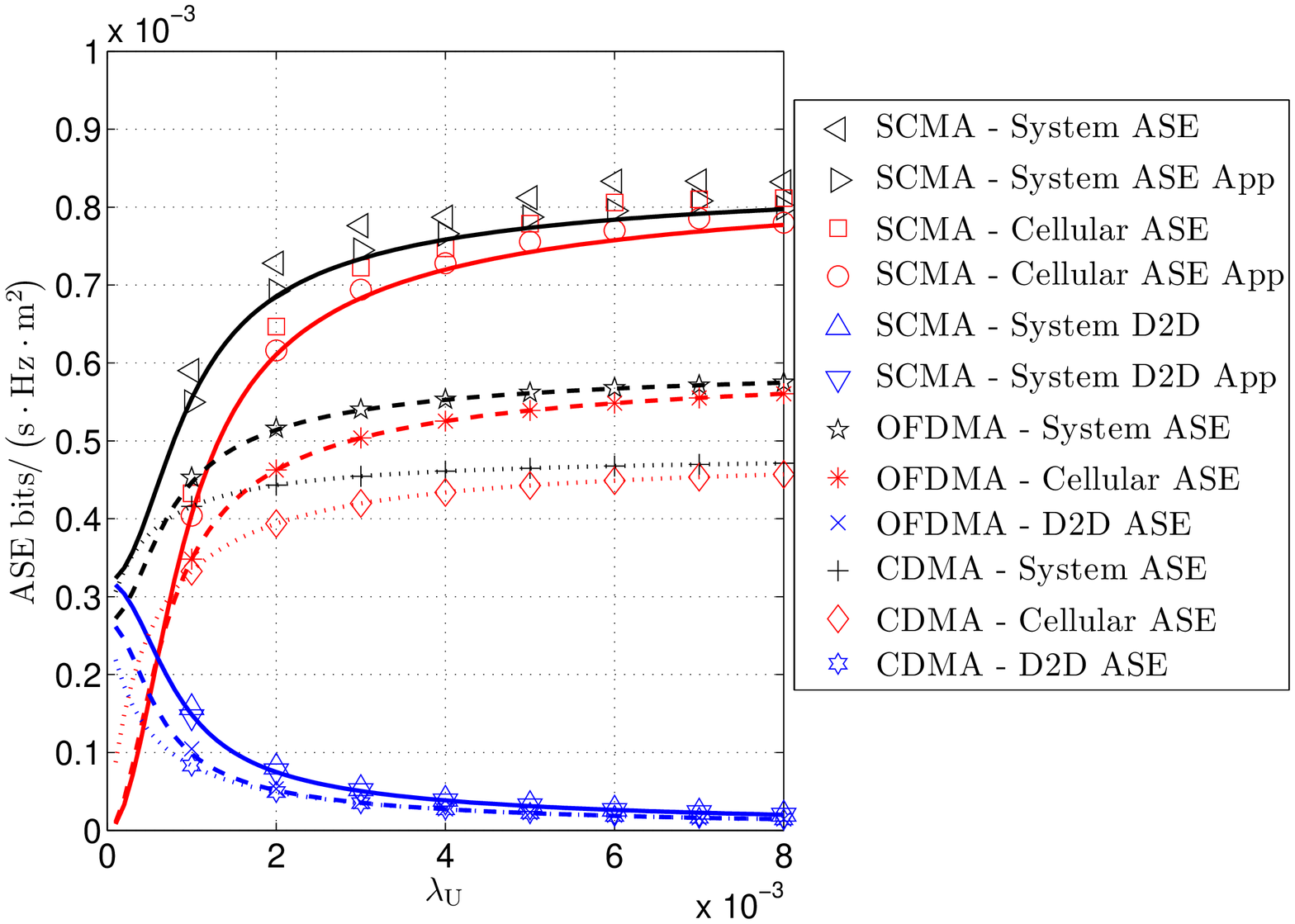}
\par\end{centering}

}\subfloat[\label{fig:OFDMA vs SCMA varying D2D intensity}Varying intensity
of D2D users.]{\begin{centering}
\includegraphics[width=3in]{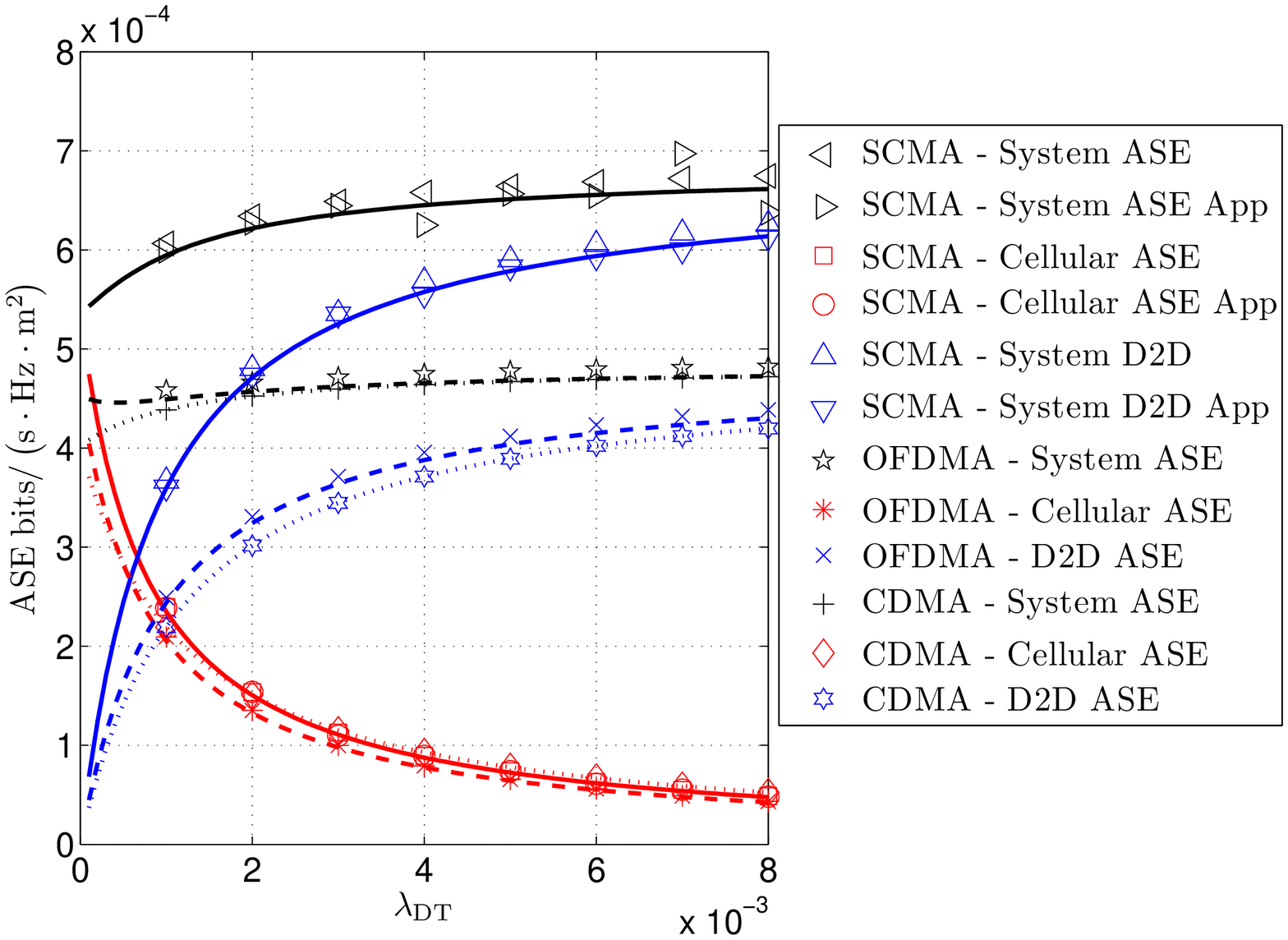}
\par\end{centering}

}
\par\end{centering}

\caption{\label{fig:OFDMA vs SCMA}ASE - CDMA, OFDMA and SCMA in the underlaid
mode. System parameters are set as $\xi=5\times10^{-5}$ and $\lambda_{\mathrm{BS}}=5\times10^{-5}\:\mathrm{BSs/\mathrm{m}^{2}}$.
In (a), $\lambda_{\mathrm{DT}}=2.5\times10^{-4}\:\mathrm{users/\mathrm{m}^{2}}$.
In (b), $\lambda_{\mathrm{U}}=1\times10^{-3}\:\mathrm{users/\mathrm{m}^{2}}$.
We set $\tau_{\mathrm{dis}}\rightarrow\infty$ so that all D2D users
select D2D mode for data transmission.}
\end{figure}

Fig. \ref{fig:OFDMA vs SCMA} shows the cellular network ASE, D2D
network ASE and system ASE when code division multiple access (CDMA),
OFDMA and SCMA are applied, respectively, as a function of the intensity
of cellular users and intensity of D2D users. Note that the ASE analysis
of the CDMA system is based on the results in \cite{Spectrum_sharing_ad_hoc}.
It is obvious from Fig. \ref{fig:OFDMA vs SCMA varying cellular intensity}
(Fig. \ref{fig:OFDMA vs SCMA varying D2D intensity}) that cellular
(D2D) network ASE first increases and then keeps stable as the intensity
of cellular (D2D) users increases. We take the case in Fig. \ref{fig:OFDMA vs SCMA varying cellular intensity}
for example. Given small $\lambda_{\mathrm{U}}$, more cellular users
can be served when $\lambda_{\mathrm{U}}$ increases. Therefore, spectrum
resources can be better exploited by admitting more cellular users
into the system. When $\lambda_{\mathrm{U}}$ is sufficiently large,
limited number of cellular users can be served due to the association
rule such that cellular network ASE hardly increases.  In this case,
SCMA significantly outperforms OFDMA due to the overloading gain harvested
by SCMA, i.e., more orthogonal resources (codebooks) can be provided
by SCMA compared to OFDMA. Therefore, more cellular users can be kept
active. Besides, we see that the performance of the CDMA network is
worse than that of the OFDMA network. The reason can be explained
as follows. Using CDMA, interference is averaged out, as CDMA signals
are spread over all the available tones. Nonetheless, the interference
becomes overwhelming as the user density increases. Hence, when cellular
users are densely deployed, it is preferable to use OFDMA, where orthogonal
resources are separately allocated to cellular users such that no
inter-user interference exists within one cell. In addition, it is
observed that little gaps exist between numerical results and simulation
results when SCMA is applied. As discussed in Section \ref{sub:OFDMA-VS-SCMA},
we use an exponentially distributed random variable to approximate
the small scale fading, which is Gamma distributed. The approximation
can provide high accuracy, as we select $N_{\mathrm{C}}=2$, which
is the typical SCMA setting.

It is shown from Fig. \ref{fig:OFDMA vs SCMA varying cellular intensity}
and Fig. \ref{fig:OFDMA vs SCMA varying D2D intensity} that the asymptotic
system ASE gains of SCMA over OFDMA can reach $\frac{\mathcal{A}_{\mathrm{C}}^{\mathrm{SU}}+\mathcal{A}_{\mathrm{D}}^{\mathrm{SU}}}{\mathcal{A}_{\mathrm{C}}^{\mathrm{OU}}+\mathcal{A}_{\mathrm{D}}^{\mathrm{OU}}}\overset{\lambda_{\mathrm{U}}\rightarrow\infty}{\longrightarrow}138.74\%$
and $\frac{\mathcal{A}_{\mathrm{C}}^{\mathrm{SU}}+\mathcal{A}_{\mathrm{D}}^{\mathrm{SU}}}{\mathcal{A}_{\mathrm{C}}^{\mathrm{OU}}+\mathcal{A}_{\mathrm{D}}^{\mathrm{OU}}}\overset{\lambda_{\mathrm{D}}\rightarrow\infty}{\longrightarrow}140.12\%$,
respectively, under the given system parameters. From the above results,
we conclude that SCMA can serve as an efficient multiple access scheme
in the D2D hybrid cellular network to enable massive connectivity.

\begin{figure}[t]
\begin{centering}
\includegraphics[width=3.5in]{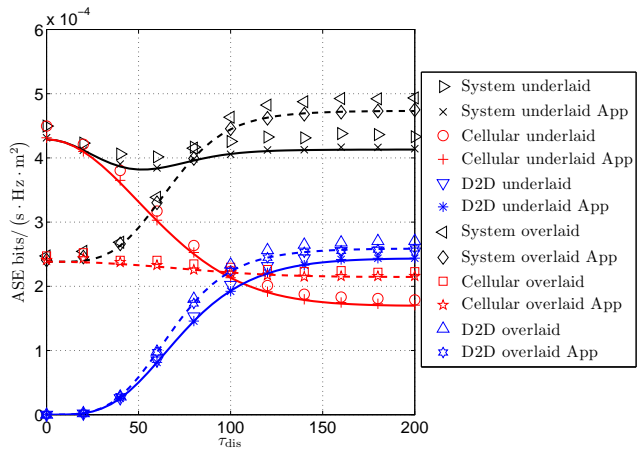}
\par\end{centering}

\caption{\label{fig:mode selection}The impact of mode selection on ASE performance
in the underlaid mode and overlaid mode. System parameters are set
as $\xi=5\times10^{-5}$, $\lambda_{\mathrm{BS}}=5\times10^{-5}\:\mathrm{BSs/\mathrm{m}^{2}}$,
$\lambda_{\mathrm{U}}=5\times10^{-4}\:\mathrm{users/\mathrm{m}^{2}}$
and $\lambda_{\mathrm{D}}=2.5\times10^{-4}\:\mathrm{users/\mathrm{m}^{2}}$.
In the overlaid mode, $J_{\mathrm{C}}=10$ and $J_{\mathrm{D}}=20$.}
\end{figure}

Fig. \ref{fig:mode selection} shows the ASE as a function of the
mode selection threshold $\tau_{\mathrm{dis}}$ in the underlaid mode
and overlaid mode. When $\tau_{\mathrm{dis}}=0\:\mathrm{m}$, all
the D2D users select cellular mode. In consequence, better ASE performance
can be achieved by the underlaid mode. This is obvious since more
codebooks are available in the underlaid mode. When $\tau_{\mathrm{dis}}$
increases, more D2D users select D2D mode for data transmission. Therefore,
it is clear from Fig. \ref{fig:mode selection} that D2D network ASE
increases with $\tau_{\mathrm{dis}}$. In contrast, cellular network
ASE decreases with $\tau_{\mathrm{dis}}$ in both coexisting modes,
since the available resources cannot be fully exploited when the number
of cellular users decreases. Besides, as the number of active D2D
transmitters increases with $\tau_{\mathrm{dis}}$, the cross-tier
interference from D2D network to cellular network becomes overwhelming
with the increasing $\tau_{\mathrm{dis}}$ in the underlaid mode.
Therefore, cellular network ASE is degraded faster in the underlaid
mode. As a result, the overlaid mode outperforms the underlaid mode
in terms of system ASE performance when $\tau_{\mathrm{dis}}$ is
large under this system setting.

\begin{figure}[t]
\begin{centering}
\includegraphics[width=3.5in]{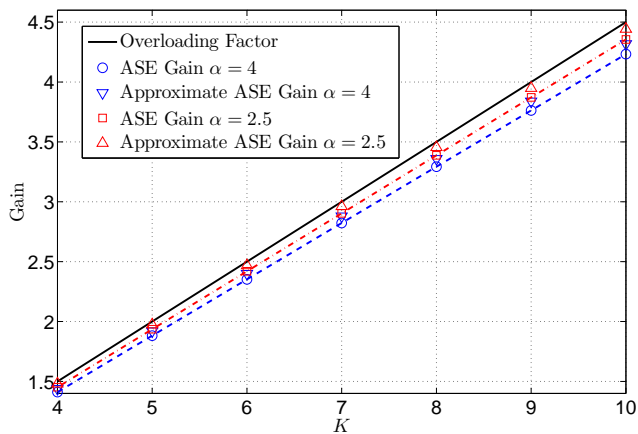}
\par\end{centering}

\caption{\label{fig:ASE gain vs overloading factor}ASE gain v.s. overloading
factor. System parameters are set as $\lambda_{\mathrm{BS}}=5\times10^{-5}\:\mathrm{BSs/\mathrm{m}^{2}}$
and $\lambda_{\mathrm{D}}=0\:\mathrm{users/\mathrm{m}^{2}}$.}

\end{figure}

Fig. \ref{fig:ASE gain vs overloading factor} compares the ASE gain
$\hat{\eta}_{\mathrm{ASE}}$ defined in (\ref{eq:ASE gain ex}) and
overloading factor $\eta_{\mathrm{overload}}$ defined in (\ref{eq:overloading factor})
as a function of the number of available OFDMA tones in the system.
In order to make a direct and comprehensive comparison, we simulate
the heavily loaded cellular network, where only cellular users exist
with $\lambda_{\mathrm{U}}\gg J\lambda_{\mathrm{BS}}$ at each $K$
and no D2D users are activated. It should be noted that we have assumed
$\mathrm{C}_{K}^{N_{\mathrm{C}}}$ codebooks can be designed for the
comparison. It can be seen that, although ASE gain is always smaller
than the overloading factor, the overloading gain can be almost achieved
even when codebook reuse is enabled. The reduction of overloading
gain primarily stems from the inter-cluster interference, which is
caused by codebook reuse in different cells.  In addition, it is
shown that the ASE gain increases linearly with $K$, which is optimistic.
Nevertheless, the result is based on the assumption that the number
of codebooks increases with the number of OFDMA tones as $J=\mathrm{C}_{K}^{N_{\mathrm{C}}}$.
This is achieved by searching for $K-1$ constellations at each $K$
according to \cite{SCMA_original,SCMA_codebook_design}. Unfortunately,
the number of constellations is limited due to practical concerns.
Consequently, the performance improvement of SCMA over OFDMA is thus
limited.

\begin{figure}[t]
\begin{centering}
\subfloat[\label{fig:MAP on ASE}System ASE.]{\begin{centering}
\includegraphics[width=3in]{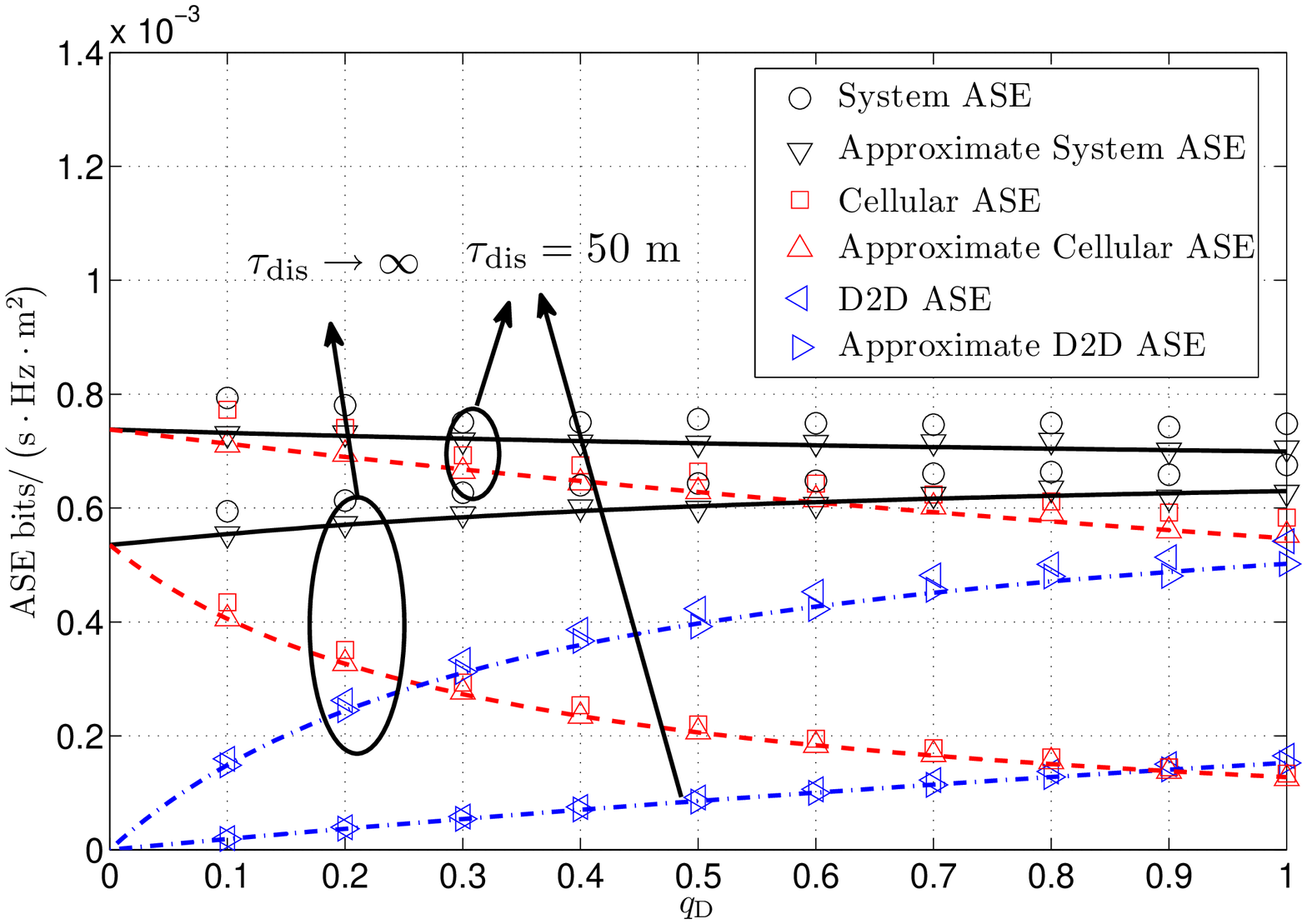}
\par\end{centering}

}\subfloat[\label{fig:Optimality of MAP.}Optimality of the activated probability.]{\begin{centering}
\includegraphics[width=3in]{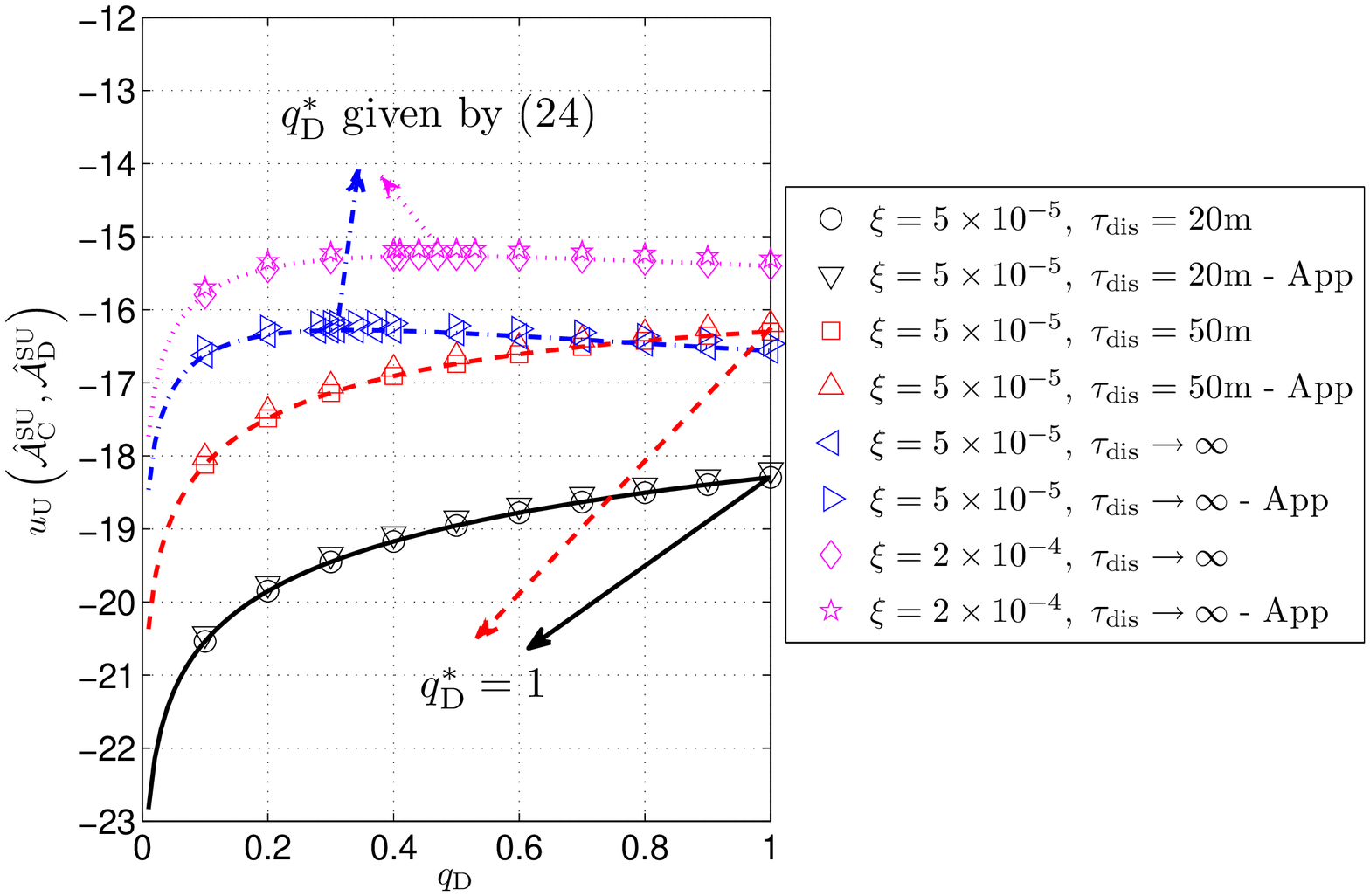}
\par\end{centering}

}
\par\end{centering}

\caption{\label{fig:D2D MAP}The impact of D2D activated probability in D2D
underlaid cellular network. System parameters are set as $\lambda_{\mathrm{BS}}=5\times10^{-5}\:\mathrm{BSs/\mathrm{m}^{2}}$,
$\lambda_{\mathrm{U}}=1\times10^{-3}\:\mathrm{users/\mathrm{m}^{2}}$
and $\lambda_{\mathrm{D}}=2.5\times10^{-3}\:\mathrm{users/\mathrm{m}^{2}}$.}
\end{figure}

We next explore the impact of D2D activated probability $q_{\mathrm{D}}$
on the performance of underlaid system in Fig. \ref{fig:D2D MAP}.
Fig. \ref{fig:MAP on ASE} shows the cellular network ASE, D2D network
ASE and system ASE as a function of $q_{\mathrm{D}}$. The tradeoff
between cellular network ASE and D2D network ASE caused by $q_{\mathrm{D}}$
is clearly presented in Fig. \ref{fig:MAP on ASE}, namely, cellular
network ASE and D2D network ASE, respectively, decreases and increases
with the increasing $q_{\mathrm{D}}$. Fig. \ref{fig:Optimality of MAP.}
shows the utility function defined in (\ref{eq:utility function underlaid})
as a function of $q_{\mathrm{D}}$. Particularly, when $\tau_{\mathrm{dis}}\rightarrow\infty$,
the utility function is an increasing function of $q_{\mathrm{D}}$
such that the optimal $q_{\mathrm{D}}$ equals 1, as indicated by
Theorem \ref{theorem: optimize MAP approximation}. Otherwise, when
$\tau_{\mathrm{dis}}$ is small, it is shown that $q_{\mathrm{D}}^{*}$
increases with $\xi$. This is because a larger $\xi$ will lead to
a smaller average D2D link length. Therefore, increasing $\xi$ would
improve $\hat{\mathcal{A}}_{\mathrm{D}}^{\mathrm{SU}}$, thereby increasing
the proportional fairness utility function.

\begin{figure}[t]
\begin{centering}
\subfloat[\label{fig:Codebook Allocation System ASE}System ASE.]{\begin{centering}
\includegraphics[width=3in]{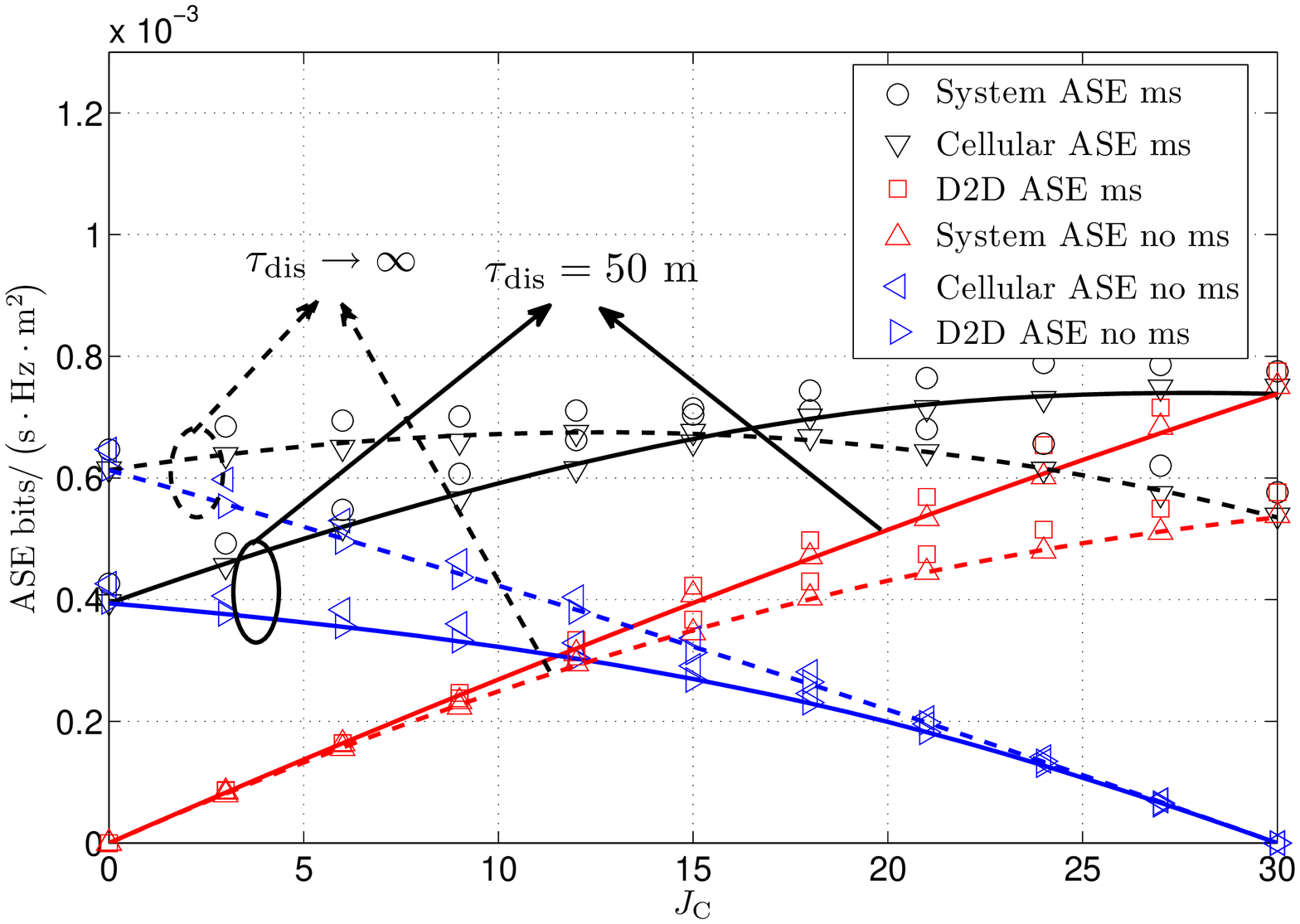}
\par\end{centering}

}\subfloat[\label{fig:Codebook allocation optimality }Optimality of codebook
allocation.]{\begin{centering}
\includegraphics[width=3in]{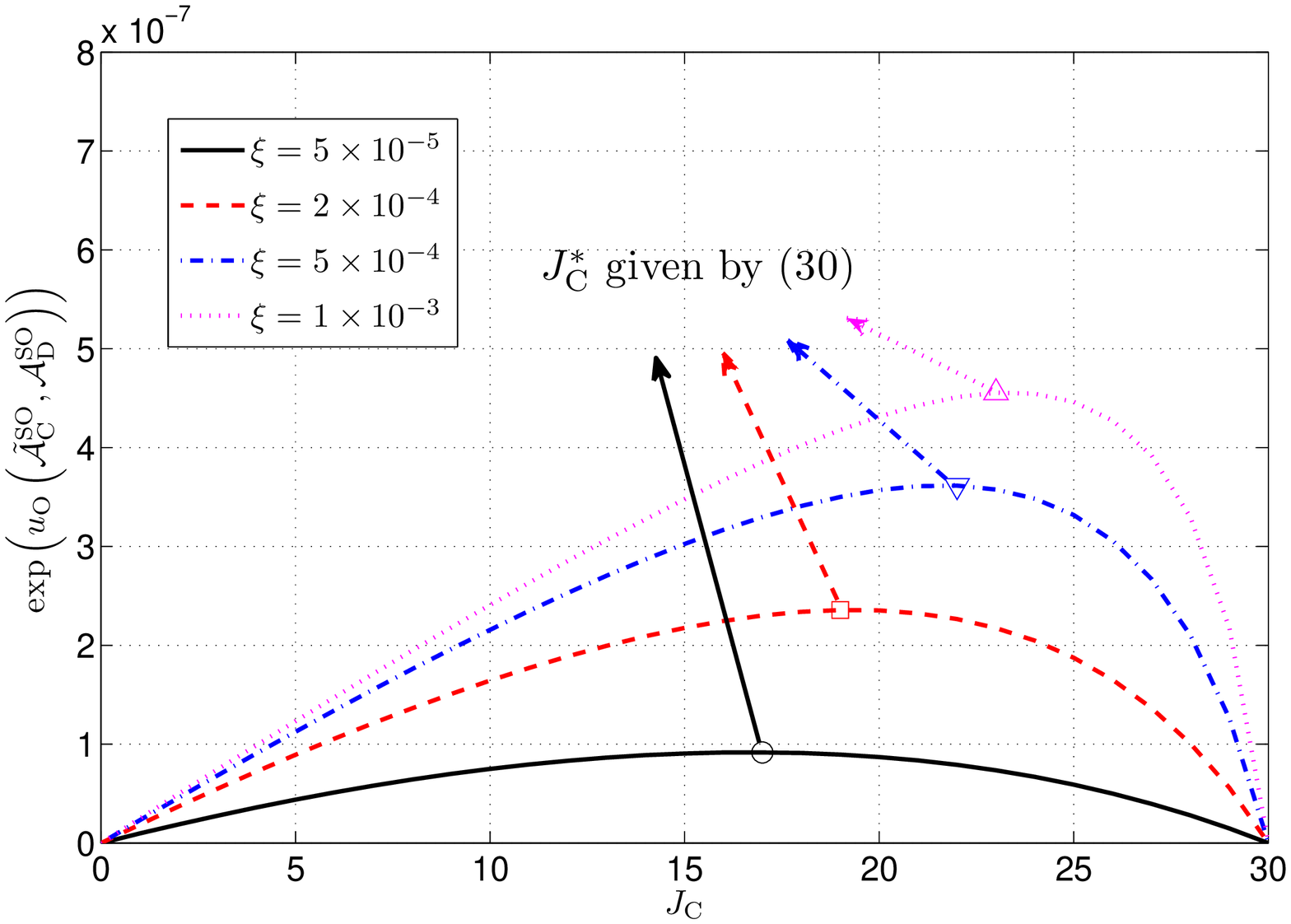}
\par\end{centering}

}
\par\end{centering}

\caption{\label{fig:Codebook Allocation Overlaid}Codebook allocation in D2D
overlaid cellular network. System parameters are set as $\lambda_{\mathrm{BS}}=5\times10^{-5}\:\mathrm{BSs/\mathrm{m}^{2}}$,
$\lambda_{\mathrm{U}}=1\times10^{-3}\:\mathrm{users/\mathrm{m}^{2}}$
and $\lambda_{\mathrm{D}}=2.5\times10^{-3}\:\mathrm{users/\mathrm{m}^{2}}$.
In (b), we set $\tau_{\mathrm{dis}}\rightarrow\infty$ such that all
D2D users select D2D mode for data transmission.}
\end{figure}

We finally investigate the performance of the D2D overlaid cellular
network when different codebook allocation schemes are applied in
Fig. \ref{fig:Codebook Allocation Overlaid}. Fig. \ref{fig:Codebook Allocation System ASE}
shows the cellular network ASE, D2D network ASE and system ASE as
a function of the number of codebooks allocated to cellular network.
It is shown that allocating more codebooks to cellular (D2D) network
would enhance the ASE of the corresponding network. This is due to
the resource allocation rule introduced in Section \ref{sub:Multiple Access Schemes}.
In particular, each cellular (D2D) user is randomly assigned with
one of the available codebooks. If more codebooks are allocated, one
codebook will be reused by less users, thereby resulting in less interference
over one codebook. Meanwhile, Fig. \ref{fig:Codebook allocation optimality }
shows the utility function defined in (\ref{eq:utility function overlaid})
as a function of the number of codebooks allocated to cellular users
$J_{\mathrm{C}}^{*}$. It is shown from the comparison between numerical
results and simulation results that $J_{\mathrm{C}}^{*}$ obtained
in Theorem \ref{theorem: optimize codebook allocation} is of high
accuracy.

\section{Conclusion\label{sec:Conclusion}}

In this paper, we have presented a stochastic geometry based framework
to investigate the performance of SCMA in D2D and cellular hybrid
network, considering underlaid mode and overlaid mode. In the underlaid
mode, we analytically compared SCMA with OFDMA from the perspective
of ASE and quantified the ASE gain of SCMA over OFDMA in closed form.
Though approximations are used, it has been shown that high accuracy
can be provided. More importantly, we concluded via the comparison
results that SCMA is capable of supporting massive D2D connections,
as well as significantly enhancing system ASE especially in the heavily
loaded network. Therefore, SCMA can be considered as a candidate of
effective multiple access schemes in future 5G wireless networks.
In addition, spectrum sharing in the two coexisting modes has been
studied. Specifically, the optimal D2D activated probability has been
derived in the underlaid mode and the optimal codebook allocation
rule has been obtained in the overlaid mode. In both cases, the optimization
targets are to maximize the proportional fairness utility function.
The results can serve as a guide to help tune system design parameters
in SCMA enhanced D2D and cellular hybrid network.

\appendix

\section{*}

\subsection{Proof for Proposition \ref{proposition:  OFDMA ASE underlaid}\label{sub:Proof for Prop cellular CP OFDMA}}

According to (\ref{eq:SIR at BS OFDMA}), we have
\begin{align}
\mathrm{CP}_{\mathrm{BS}}^{\mathrm{OU}} & =\mathbb{P}\left\{ \frac{P_{\mathrm{U}}r_{\mathrm{U},0}^{-\alpha}\left\Vert h_{\mathrm{BS}_{0},\mathrm{U}_{0}}\right\Vert ^{2}}{I_{\mathrm{C},\mathrm{C}}^{\mathrm{O}}+I_{\mathrm{C},\mathrm{D}}^{\mathrm{O}}}>\tau_{\mathrm{BS}}\right\} =\mathbb{P}\left\{ \left\Vert h_{\mathrm{BS}_{0},\mathrm{U}_{0}}\right\Vert ^{2}>\frac{\tau_{\mathrm{BS}}r_{\mathrm{U},0}^{\alpha}}{P_{\mathrm{U}}}\left(I_{\mathrm{C},\mathrm{C}}^{\mathrm{O}}+I_{\mathrm{C},\mathrm{D}}^{\mathrm{O}}\right)\right\} \nonumber \\
 & \overset{\left(\mathrm{a}\right)}{=}\mathcal{L}_{I_{\mathrm{C},\mathrm{C}}^{\mathrm{O}}}\left(s_{\mathrm{BS}}^{\mathrm{O}}\right)\mathcal{L}_{I_{\mathrm{C},\mathrm{D}}^{\mathrm{O}}}\left(s_{\mathrm{BS}}^{\mathrm{O}}\right),\label{eq:CP at BS OFDMA proof}
\end{align}
where (a) follows due to $\left\Vert h_{\mathrm{BS}_{0},\mathrm{U}_{0}}\right\Vert ^{2}\sim\exp\left(1\right)$.
$\mathcal{L}_{I_{\mathrm{C},\mathrm{C}}^{\mathrm{O}}}\left(s_{\mathrm{BS}}^{\mathrm{O}}\right)$
and $\mathcal{L}_{I_{\mathrm{C},\mathrm{D}}^{\mathrm{O}}}\left(s_{\mathrm{BS}}^{\mathrm{O}}\right)$,
respectively, denote the Laplace transforms of $I_{\mathrm{C},\mathrm{C}}^{\mathrm{O}}$
and $I_{\mathrm{C},\mathrm{D}}^{\mathrm{O}}$ evaluated at $s_{\mathrm{BS}}^{\mathrm{O}}=\frac{\tau_{\mathrm{BS}}r_{\mathrm{U},0}^{\alpha}}{P_{\mathrm{U}}}$.
We evaluate $\mathcal{L}_{I_{\mathrm{C},\mathrm{C}}^{\mathrm{O}}}\left(s_{\mathrm{BS}}^{\mathrm{O}}\right)$
as
\begin{align}
\mathcal{L}_{I_{\mathrm{C},\mathrm{C}}^{\mathrm{O}}}\left(s_{\mathrm{BS}}^{\mathrm{O}}\right) & \overset{\left(\mathrm{a}\right)}{=}\mathbb{E}\left[\underset{\tiny{\mathrm{U}_{i}\in\tilde{\Pi}_{\mathrm{UT}}^{\mathrm{O}}}}{\prod}\frac{1}{1+s_{\mathrm{BS}}^{\mathrm{O}}P_{\mathrm{U}}\left\Vert \mathrm{U}_{i}-\mathrm{BS}_{0}\right\Vert ^{-\alpha}}\right]\nonumber \\
 & \overset{\left(\mathrm{b}\right)}{=}\exp\left[-2\pi\frac{q_{\mathrm{U}}^{\mathrm{O}}\lambda_{\mathrm{UT}}}{K}\int_{r_{\mathrm{U},0}}^{\infty}l\left(1-\frac{1}{1+s_{\mathrm{BS}}^{\mathrm{O}}P_{\mathrm{U}}l^{-\alpha}}\right)dl\right]\nonumber \\
 & =\exp\left[-\frac{2\pi q_{\mathrm{U}}^{\mathrm{O}}\lambda_{\mathrm{UT}}s_{\mathrm{BS}}^{\mathrm{O}}P_{\mathrm{U}}}{K\left(\alpha-2\right)r_{\mathrm{U},0}^{\alpha-2}}{}_{2}F_{1}\left(1,1-\delta,2-\delta,-r_{\mathrm{U},0}^{-\alpha}s_{\mathrm{BS}}^{\mathrm{O}}P_{\mathrm{U}}\right)\right],\label{eq:Laplace of ICC OFDMA}
\end{align}
where (a) follows due to $\left\Vert h_{\mathrm{BS}_{0},\mathrm{U}_{i}}\right\Vert ^{2}\sim\exp\left(1\right)$
and (b) follows due to the probability generating functional (PGFL)
of PPP and cellular association rule that $\mathrm{U}_{0}$ connects
to the nearest BS so that interfering cellular users are at least
$r_{0}$ away. Therefore, $\mathcal{L}_{I_{\mathrm{C},\mathrm{C}}^{\mathrm{O}}}\left(\frac{\tau_{\mathrm{BS}}r_{\mathrm{U},0}^{\alpha}}{P_{\mathrm{U}}}\right)$
can be obtained by substituting $s_{\mathrm{BS}}^{\mathrm{O}}=\frac{\tau_{\mathrm{BS}}r_{\mathrm{U},0}^{\alpha}}{P_{\mathrm{U}}}$
into (\ref{eq:Laplace of ICC OFDMA}). Following similar approach,
it is easy to derive $\mathcal{L}_{I_{\mathrm{C},\mathrm{D}}^{\mathrm{O}}}\left(\frac{\tau_{\mathrm{BS}}r_{\mathrm{U},0}^{\alpha}}{P_{\mathrm{U}}}\right)=\exp\left[-\frac{2\pi^{2}\lambda_{\mathrm{DT}}\left(\frac{\tau_{\mathrm{BS}}P_{\mathrm{D}}}{P_{\mathrm{U}}}\right)^{\delta}r_{\mathrm{U},0}^{2}}{K\alpha\sin\left(\frac{2\pi}{\alpha}\right)}\right]$.

Due to the association rule of cellular users, the CDF of $r_{\mathrm{U},0}$
can be obtained according to the contact distribution \cite{book_stochastic_geometry},
i.e., $F_{r_{\mathrm{U},0}}\left(x\right)=1-\exp\left(-\pi\lambda_{\mathrm{BS}}x^{2}\right)$
$\left(x\geq0\right)$. Thus, the PDF of $r_{\mathrm{U},0}$ can be
accordingly obtained as $f_{r_{\mathrm{U},0}}\left(x\right)=2\pi\lambda_{\mathrm{BS}}x\exp\left(-\pi\lambda_{\mathrm{BS}}x^{2}\right)$
$\left(x\geq0\right)$. Combining $\mathcal{L}_{I_{\mathrm{C},\mathrm{C}}^{\mathrm{O}}}\left(\frac{\tau_{\mathrm{BS}}r_{\mathrm{U},0}^{\alpha}}{P_{\mathrm{U}}}\right)$
and $\mathcal{L}_{I_{\mathrm{C},\mathrm{D}}^{\mathrm{O}}}\left(\frac{\tau_{\mathrm{BS}}r_{\mathrm{U},0}^{\alpha}}{P_{\mathrm{U}}}\right)$
and taking the expectation of $r_{\mathrm{U},0}$, we derive the coverage
probability at the typical cellular BS. Using similar approach, we
have
\begin{equation}
\mathrm{CP}_{\mathrm{DR}}^{\mathrm{OU}}=\exp\left[-\frac{2\pi^{2}\tau_{\mathrm{DR}}^{\delta}r_{\mathrm{D},0}^{2}}{K\alpha\sin\left(\frac{2\pi}{\alpha}\right)}\left(\lambda_{\mathrm{DT}}+q_{\mathrm{U}}^{\mathrm{O}}\lambda_{\mathrm{UT}}\left(\frac{P_{\mathrm{U}}}{P_{\mathrm{D}}}\right)^{\delta}\right)\right].\label{eq:CP at D2D OFDMA proof}
\end{equation}
By taking the expectation of $r_{\mathrm{D},0}$ in (\ref{eq:CP at D2D OFDMA proof})
according to the PDF given in (\ref{eq:D2D link length distribution}),
we derive the coverage probability at the typical D2D receiver. Note
that $r_{\mathrm{D},0}\in\left[0,\tau_{\mathrm{dis}}\right]$, since
mode selection is applied by D2D transmitters.

\subsection{Proof for Proposition \ref{proposition: SCMA ASE underlaid}\label{sub:Proof for Prop cellular CP SCMA}}

We denote $\underset{m=1}{\overset{N_{\mathrm{C}}}{\sum}}\left\Vert h_{\mathrm{BS}_{0},\mathrm{U}_{0},m}\right\Vert ^{2}$
by $G_{\mathrm{BS}_{0},\mathrm{U}_{0}}$. Since $\left\Vert h_{\mathrm{BS}_{0},\mathrm{U}_{0},m}\right\Vert ^{2}\sim\exp\left(1\right)$,
$G_{\mathrm{BS}_{0},\mathrm{U}_{0}}$ follows Gamma distribution,
i.e., $G_{\mathrm{BS}_{0},\mathrm{U}_{0}}\sim\mathrm{Gamma}\left(N_{\mathrm{C}},1\right)$,
which makes it difficult to derive the explicit-form result of $\mathrm{CP}_{\mathrm{BS}}^{\mathrm{SU}}$.
In order to provide analytical tractability, we use a random variable
$W_{\mathrm{BS}_{0},\mathrm{U}_{0}}$ with mean $N_{\mathrm{C}}$,
i.e., $W_{\mathrm{BS}_{0},\mathrm{U}_{0}}\sim\exp\left(N_{\mathrm{C}}^{-1}\right)$
to approximate $G_{\mathrm{BS}_{0},\mathrm{U}_{0}}$. Note that $G_{\mathrm{BS}_{0},\mathrm{U}_{0}}$
and $W_{\mathrm{BS}_{0},\mathrm{U}_{0}}$ have the same first moment.
Therefore, we have the approximation
\begin{align}
\mathrm{CP}_{\mathrm{BS}}^{\mathrm{SU}} & \approx\mathbb{P}\left\{ \frac{P_{\mathrm{U}}^{\dagger}r_{\mathrm{U},0}^{-\alpha}W_{\mathrm{BS}_{0},\mathrm{U}_{0}}}{I_{\mathrm{C},\mathrm{C}}^{\mathrm{S}}+I_{\mathrm{C},\mathrm{D}}^{\mathrm{S}}}>\tau_{\mathrm{BS}}\right\} =\mathcal{L}_{I_{\mathrm{C},\mathrm{C}}^{\mathrm{S}}}\left(s_{\mathrm{BS}}^{\mathrm{S}}\right)\mathcal{L}_{I_{\mathrm{C},\mathrm{D}}^{\mathrm{S}}}\left(s_{\mathrm{BS}}^{\mathrm{S}}\right),\label{eq:CP at BS SCMA proof}
\end{align}
where $\mathcal{L}_{I_{\mathrm{C},\mathrm{C}}^{\mathrm{S}}}\left(s_{\mathrm{BS}}^{\mathrm{S}}\right)$
and $\mathcal{L}_{I_{\mathrm{C},\mathrm{D}}^{\mathrm{S}}}\left(s_{\mathrm{BS}}^{\mathrm{S}}\right)$,
respectively, denote the Laplace transforms of $I_{\mathrm{C},\mathrm{C}}^{\mathrm{S}}$
and $I_{\mathrm{C},\mathrm{D}}^{\mathrm{S}}$ evaluated at $s_{\mathrm{BS}}^{\mathrm{S}}=\frac{\tau_{\mathrm{BS}}r_{\mathrm{U},0}^{\alpha}}{N_{\mathrm{C}}P_{\mathrm{U}}^{\dagger}}$.
We denote $\overset{N_{\mathrm{C}}}{\underset{m=1}{\sum}}\left\Vert h_{\mathrm{BS}_{0},\mathrm{U}_{i},m}\right\Vert ^{2}$
as $G_{\mathrm{BS}_{0},\mathrm{U}_{i}}$ and calculate $\mathcal{L}_{I_{\mathrm{C},\mathrm{C}}^{\mathrm{S}}}\left(s_{\mathrm{BS}}^{\mathrm{S}}\right)$
as
\begin{align}
\mathcal{L}_{I_{\mathrm{C},\mathrm{C}}^{\mathrm{S}}}\left(s_{\mathrm{BS}}^{\mathrm{S}}\right) & =\mathbb{E}\left[\exp\left(-\underset{\tiny{\mathrm{U}_{i}\in\tilde{\Pi}_{\mathrm{UT}}^{\mathrm{S}}}}{\sum}s_{\mathrm{BS}}^{\mathrm{S}}P_{\mathrm{U}}^{\dagger}G_{\mathrm{BS}_{0},\mathrm{U}_{i}}\left\Vert \mathrm{U}_{i}-\mathrm{BS}_{0}\right\Vert ^{-\alpha}\right)\right]\nonumber \\
 & \overset{\left(\mathrm{a}\right)}{=}\mathbb{E}\left[\underset{\tiny{\mathrm{U}_{i}\in\tilde{\Pi}_{\mathrm{UT}}^{\mathrm{S}}}}{\prod}\frac{1}{\left(1+s_{\mathrm{BS}}^{\mathrm{S}}P_{\mathrm{U}}^{\dagger}\left\Vert \mathrm{U}_{i}-\mathrm{BS}_{0}\right\Vert ^{-\alpha}\right)^{N_{\mathrm{C}}}}\right]\nonumber \\
 & =\exp\left[-2\pi\frac{q_{\mathrm{U}}^{\mathrm{S}}\lambda_{\mathrm{UT}}}{J}\int_{r_{\mathrm{U},0}}^{\infty}l\left(1-\frac{1}{\left(1+s_{\mathrm{BS}}^{\mathrm{S}}P_{\mathrm{U}}^{\dagger}l^{-\alpha}\right)^{N_{\mathrm{C}}}}\right)dl\right]\nonumber \\
 & =\exp\left[-\frac{\pi q_{\mathrm{U}}^{\mathrm{S}}\lambda_{\mathrm{UT}}}{r_{\mathrm{U},0}^{-2}J}\left(_{2}F_{1}\left(N_{\mathrm{C}},-\delta,1-\delta,-r_{\mathrm{U},0}^{-\alpha}s_{\mathrm{BS}}^{\mathrm{S}}P_{\mathrm{U}}^{\dagger}\right)-1\right)\right],\label{eq:Laplace of ICC SCMA}
\end{align}
where (a) is due to the PGFL of PPP. Likewise, we calculate $\mathcal{L}_{I_{\mathrm{C},\mathrm{D}}^{\mathrm{S}}}\left(s_{\mathrm{BS}}^{\mathrm{S}}\right)$
as 
\begin{align}
\mathcal{L}_{I_{\mathrm{C},\mathrm{D}}^{\mathrm{S}}}\left(s_{\mathrm{BS}}^{\mathrm{S}}\right) & =\mathbb{E}\left[\exp\left(-\underset{\tiny{\mathrm{DT}_{j}\in\Pi_{\mathrm{DT}}^{\mathrm{S}}}}{\sum}s_{\mathrm{BS}}^{\mathrm{S}}P_{\mathrm{D}}^{\dagger}G_{\mathrm{BS}_{0},\mathrm{DT}_{i}}\left\Vert \mathrm{DT}_{j}-\mathrm{BS}_{0}\right\Vert ^{-\alpha}\right)\right]\nonumber \\
 & =\mathbb{E}\left[\underset{\tiny{\mathrm{DT}_{j}\in\Pi_{\mathrm{DT}}^{\mathrm{S}}}}{\prod}\frac{1}{\left(1+s_{\mathrm{BS}}^{\mathrm{S}}P_{\mathrm{D}}^{\dagger}\left\Vert \mathrm{DT}_{i}-\mathrm{BS}_{0}\right\Vert ^{-\alpha}\right)^{N_{\mathrm{C}}}}\right]\nonumber \\
 & \overset{\left(\mathrm{a}\right)}{=}\exp\left[-2\pi\frac{\lambda_{\mathrm{DT}}}{J}\int_{0}^{\infty}l\left(1-\frac{1}{\left(1+s_{\mathrm{BS}}^{\mathrm{S}}P_{\mathrm{D}}^{\dagger}l^{-\alpha}\right)^{N_{\mathrm{C}}}}\right)dl\right]\nonumber \\
 & =\exp\left[-\frac{2\pi^{2}\lambda_{\mathrm{DT}}\left(s_{\mathrm{BS}}^{\mathrm{S}}P_{\mathrm{D}}^{\dagger}\right)^{\delta}}{J\alpha\sin\left(\frac{2\pi}{\alpha}\right)}\underset{_{n=2}}{\overset{_{N_{\mathrm{C}}}}{\prod}}\left(\frac{2}{\left(n-1\right)\alpha}+1\right)\right],\label{eq:Laplace of ICD SCMA}
\end{align}
where (a) follows because all D2D transmitters are active.

Substituting $s_{\mathrm{BS}}^{\mathrm{S}}=\frac{\tau_{\mathrm{BS}}r_{\mathrm{U},0}^{\alpha}}{N_{\mathrm{C}}P_{\mathrm{U}}^{\dagger}}$
into (\ref{eq:Laplace of ICC SCMA}) and (\ref{eq:Laplace of ICD SCMA}),
we have

\begin{equation}
\mathcal{L}_{I_{\mathrm{C},\mathrm{C}}^{\mathrm{S}}}\left(\frac{\tau_{\mathrm{BS}}r_{\mathrm{U},0}^{\alpha}}{N_{\mathrm{C}}P_{\mathrm{U}}^{\dagger}}\right)=\exp\left[-\pi\frac{q_{\mathrm{U}}^{\mathrm{S}}\lambda_{\mathrm{UT}}}{J}r_{\mathrm{U},0}^{2}\left(_{2}F_{1}\left(N_{\mathrm{C}},-\delta,1-\delta,-\frac{\tau_{\mathrm{BS}}}{N_{\mathrm{C}}}\right)-1\right)\right].\label{eq:Laplace of ICC SCMA ex}
\end{equation}
\begin{equation}
\mathcal{L}_{I_{\mathrm{C},\mathrm{D}}^{\mathrm{S}}}\left(\frac{\tau_{\mathrm{BS}}r_{\mathrm{U},0}^{\alpha}}{N_{\mathrm{C}}P_{\mathrm{U}}^{\dagger}}\right)=\exp\left[-\frac{2\pi^{2}\lambda_{\mathrm{DT}}\left(\frac{\tau_{\mathrm{BS}}P_{\mathrm{D}}^{\dagger}}{N_{\mathrm{C}}P_{\mathrm{U}}^{\dagger}}\right)^{\delta}r_{\mathrm{U},0}^{2}}{J\alpha\sin\left(\frac{2\pi}{\alpha}\right)}\underset{_{n=2}}{\overset{_{N_{\mathrm{C}}}}{\prod}}\left(\frac{2}{\left(n-1\right)\alpha}+1\right)\right].\label{eq:Laplace of ICD SCMA ex}
\end{equation}

The PDF of $r_{\mathrm{U},0}$ can be obtained in Appendix \ref{sub:Proof for Prop cellular CP OFDMA}.
Therefore, we derive the approximate coverage probability at the typical
BS by combining the results of (\ref{eq:Laplace of ICC SCMA ex})
and (\ref{eq:Laplace of ICD SCMA ex}) into (\ref{eq:CP at BS SCMA proof})
and taking the expectation of $r_{\mathrm{U},0}$.

Using similar approach, we approximate the coverage probability at
the typical D2D receiver and then complete the proof. The detail of
derivation steps is omitted due to space limitation.

\bibliographystyle{IEEEtran}
\bibliography{ref}

\end{document}